\documentclass[12pt, letterpaper]{amsart}

\newcommand{\la}{\langle}
\newcommand{\ra}{\rangle}

\newcommand{\supp}{\operatorname{supp}}

\usepackage{amssymb}
\usepackage{amsthm}
\usepackage{amsxtra}
\usepackage{graphicx}
\usepackage{color}
\usepackage{xcolor}
\usepackage{mathrsfs}

\newtheorem{theorem}{Theorem}
\newtheorem{definition}[theorem]{Definition}

\newtheorem{lemma}[theorem]{Lemma}
\newtheorem{corollary}[theorem]{Corollary}

\theoremstyle{remark}

\numberwithin{equation}{section}

\numberwithin{theorem}{section}

\numberwithin{table}{section}

\numberwithin{figure}{section}

\def\Xd{X^\dagger}
\def\Vd{V^\dagger}
\def\Yd{Y^\dagger}
\def\Wd{W^\dagger}

\def\Id{I^\dagger}

\def\dXd{\dot{X}^\dagger}
\def\dVd{\dot{V}^\dagger}

\def\hVd{\hat{V}^\dagger}
\def\hWd{\hat{W}^\dagger}

\ifx\pdfoutput\undefined
  \DeclareGraphicsExtensions{.pstex, .eps}
\else
  \ifx\pdfoutput\relax
    \DeclareGraphicsExtensions{.pstex, .eps}
  \else
    \ifnum\pdfoutput>0
      \DeclareGraphicsExtensions{.pdf}
    \else
      \DeclareGraphicsExtensions{.pstex, .eps}
    \fi
  \fi
\fi

\title[Magnetic confinement approximates specular reflection]{Magnetic confinement at a boundary approximates specular reflection}

\author{Katherine Zhiyuan Zhang}
\address{Courant Institute of Mathematical Sciences, New York University}

\begin{document}

\maketitle

\begin{abstract}
We conjecture that for a plasma in a spatial domain with a boundary, the specular reflection effect of the boundary can be approximated by a large magnetic confinement field in the near-boundary region. In this paper, we verify this conjecture for the 1.5D relativistic Vlasov-Maxwell system (RVM) on a bounded domain $\Omega = (0, 1)$ with an external confining magnetic field.
\end{abstract}

\section{Introduction} \label{Intro}

It has been one of the major goals of fusion energy research to confine plasmas (charged fluids). Scientists are particularly interested in designing stable devices to induce confinement. An external \emph{confining} magnetic field (a "magnetic mirror/shield") is one of the effective tools for this goal. The mathematical justification of this confining mechanism for plasma models like the Vlasov-Poisson system (VP) or the relativistic Vlasov-Maxwell (RVM) system has been carried out in various literatures under some different settings, for example, see \cite{CCM1}, \cite{CCM2}, \cite{CCM3}, \cite{NS3}, etc.. In these literatures, it is shown that the external confining magnetic field has a "reflective" effect on charged particles, which resembles the role of a specular reflecting wall. Moreover, if the confining field is strong enough near the spatial boundary (namely, blows up to $\infty$ at the boundary), then the particles do not touch the spatial boundary in any finite time interval if their initial positions are away from it (see, for example, \cite{CCM1}, \cite{CCM2}, \cite{CCM3}, \cite{NS3}). Even if the confining field is finite, as long as it is strong enough near the spatial boundary, then it still prevents the particles from touching the spatial boundary in some finite (but not necessarily small) time interval if their initial positions are away from it.


On the other hand, when considering kinetic models for plasmas on a domain $\Omega$ with a boundary $\partial \Omega$, one of the common choices for boundary conditions on the particle density distribution function $f$ is the \emph{specular boundary condition}, which says that each particle hitting the $\partial \Omega$ gets reflected in a natural way without losing its energy:
\begin{equation*}
f (t,x,v) = f  (t,x,v-2(v \cdot e_n (x))e_n(x) ), e_n(x) \cdot v <0,  \  \forall x \in \partial \Omega , 
\end{equation*}
where $e_n (x)$ is the outward normal unit vector at $x \in \partial \Omega$. The well-posedness and stability of RVM or VP has been studied in quite a number of literatures, for example, \cite{GlasseyS1}, \cite{GlasseySchaeffer1}, \cite{G1}, \cite{G2}, \cite{HSchaeffer1}, \cite{H1}, \cite{HV1}, \cite{HV2}. In particular, in \cite{G1}, \cite{G2}, \cite{HSchaeffer1}, \cite{H1}, \cite{HV1}, \cite{HV2} the spatial boundary is taken into account and the specular boundary condition is considered.
 

We conjecture that an external magnetic confinement which is sufficiently large near the boundary provides a good approximation to the specular boundary condition for a charged fluid. In this paper, we initiate the mathematical verification of this conjecture by studying a lower-dimensional RVM model for the sake of simplicity. Upon verification of this conjecture, we are able to justify the significance of the specular boundary condition in the study of kinetic models for plasmas, since it is a effective approximation for the scenario when a magnetic mirror/shield is applied to confine a plasma in a bounded region.

We consider the relativistic Vlasov-Maxwell (RVM) system in a bounded interval $\Omega = (0, 1)$ with the time $t \geq 0$, the spatial variable $x \in \Omega$, and the particle momentum $v \in \mathbb{R}^2$, as well as an external magnetic field $B_{ext, N}$ given by
\begin{equation}  
\begin{split}
& B_{ext, N} (x) := N b(Nx) \text{ for } x \text{ near } 0 , \\
& B_{ext, N} (x) := -N b(N(1-x)) \text{ for } x \text{ near } 1 . \\
\end{split}
\end{equation}
where $b(x)$ is some $C^1$, piecewise $C^3$, compactly supported function on $(0, +\infty)$ that blows up to $+\infty$ or $-\infty$ when $x \rightarrow 0$ (see Section \ref{SectionSetup} for details). This is an 1.5D model, which is the model of lowest dimension that includes magnetic effects. We will take $N \rightarrow +\infty$ and study the effect of the external magnetic field $B_{ext, N}$ in the limit regime.

 



In the main part of the paper, we consider a plasma with a single species of particles (ion) with a non-negative distribution function $f(t, x, v)$, where $t \geq 0$, $x \in \Omega$, $v \in \mathbb{R}^2$. The Vlasov equation is
\begin{equation}  \label{VlasovBext}
\partial_t f + \hat{v}_1 \partial_x f + ( E_1 + \hat{v}_2 B + \hat{v}_2 B_{ext, N} ) \partial_{v_1} f + ( E_2 - \hat{v}_1 B - \hat{v}_1 B_{ext, N} ) \partial_{v_2} f =0  \ , 
\end{equation} 
The electromagnetic field $E= E(t, x)$, $B= B(t, x)$ satisfies the 1.5D Maxwell system
\begin{equation}  \label{MaxwellBext}
\begin{split}
& \partial_t E_1=  - j_1 \ , \ \partial_x E_1 = \rho \ , \\
& \partial_t E_2 = - \partial_x B - j_2 \ , \\
& \partial_t B = - \partial_x E_2 \ . \\ 
\end{split}
\end{equation} 
Here $\hat{v} := v/ \la v \ra$, where $ \la v \ra :=  \sqrt{1 + |v|^2}$. The charge density $\rho (t, x) := \int_{\mathbb{R}^2} f(t, x, v) dv$, and the current density $\textbf{j} (t, x) :=  \int_{\mathbb{R}^2} \hat{v} f(t, x , v) dv$. We have normalized the speed of light as well as the unit mass and charge of the particles to be $1$, since these quantities play no role in our qualitative analysis.

For the system \eqref{VlasovBext} -- \eqref{MaxwellBext}, we put down some initial data for $f$, which is supported away from $\partial \Omega$, and some appropriate initial-boundary condition (which enjoys some smoothness) for $E$ and $B$, see \eqref{boundarycondBext} in Section \ref{SectionSetup}.  


\cite{NS3} gives the global well-posedness and the $C^1$ regularity of the (strong) solution for the system \eqref{VlasovBext} and \eqref{MaxwellBext} with the initial-boundary conditions \eqref{boundarycondBext}. Moreover, the particles will not hit the boundary, due to the confining property of $B_{ext,N}$ (See Lemma 3.1 in \cite{NS3} and Lemma \ref{NS3-lm3.1} below in this paper). Therefore no boundary condition on $f$ is needed for \eqref{VlasovBext}. In Lemma \ref{NS3-lm3.1-finite}, we also discuss the case when $B_{ext, N}$ is a {\it{finite}} external magnetic confining field and prove that if $B_{ext,N}$ is large enough (depending on the initial-boundary data and the time interval $[0, T]$), then the particles will stay away from the boundary on $[0, T]$.



Our goal is to investigate the limiting behavior of the solution for \eqref{VlasovBext} -- \eqref{MaxwellBext} as $N \rightarrow + \infty$. 
To this end, we consider the 1.5D RVM on $\Omega$ with no external magnetic field. 
The Vlasov equation and the Maxwell system are
\begin{equation}  \label{VlasovBextspecular}
\partial_t f + \hat{v}_1 \partial_x f + ( E_1 + \hat{v}_2 B ) \partial_{v_1} f + ( E_2 - \hat{v}_1 B   ) \partial_{v_2} f =0  \ , 
\end{equation} 
\begin{equation}  \label{MaxwellBextspecular}
\begin{split}
& \partial_t E_1=  - j_1 \ , \ \partial_x E_1 = \rho \ , \\
& \partial_t E_2 = - \partial_x B - j_2 \ , \\
& \partial_t B = - \partial_x E_2 \ . \\ 
\end{split}
\end{equation}   
with the initial-boundary conditions \eqref{boundarycondBext} together with the specular boundary condition in the 1.5D model on the domain $\Omega = (0, 1)$:
\begin{equation}  \label{specularBC}
f (t, x, v_1, v_2) = f(t, x, -v_1, v_2 ), \ \text{for} \ x = 0, \ 1 . 
\end{equation}
Notice that without the external field the particles may hit $\partial \Omega$ so the specification of this boundary condition is necessary.  

We wish to prove that as $N \rightarrow + \infty$, the solutions for the system \eqref{VlasovBext} -- \eqref{MaxwellBext} converge to the ones for the system \eqref{VlasovBextspecular} -- \eqref{MaxwellBextspecular} with the specular boundary condition \eqref{specularBC}. That is to say, as $N \rightarrow + \infty$, the external confining magnetic field well approximates a perfectly reflecting boundary wall. This is verified in a weak sense as stated in our main result as follows:

\begin{theorem} \label{mtheorem}
For each $N$ and any $T>0$, we consider the strong $C^1$ solution $(f^N, E^N_1, E^N_2, B^N)$ on $[0, T]$ to \eqref{VlasovBext}, \eqref{MaxwellBext}, with the initial-boundary condition \eqref{boundarycondBext}. There exists a subsequence of $(f^N , E^N_1 , E^N_2 , B^N)$, such that $f^N \rightharpoonup f \ weakly^*$ in $L^\infty ([0, T] \times \Omega \times \mathbb{R}^2)$, $(E^N_1, E^N_2, B^N) \rightarrow (E_1, E_2, B)$ strongly in $ C^0  ([0, T] \times \Omega ) $. The limit $(f, E_1, E_2, B)$ is a weak solution of \eqref{VlasovBextspecular}, \eqref{MaxwellBextspecular}, with exactly the same initial and boundary conditions \eqref{boundarycondBext} and the specular boundary condition \eqref{specularBC} on $[0, T]$ (in the sense of Definition \ref{D:weak}, see Section \ref{SectionSetup}). 
\end{theorem}

The result can be extended to the following two cases: 1) The case when the external magnetic field $B_{ext,N}$ is finite; 2) The case when the plasma contains both ions and electrons. We discuss them in Section \ref{SectionFinite} and Section \ref{SectionTwoSpecies}, respectively. Moreover, the same result obviously holds when the spatial domain $(0, 1)$ is replaced by a half line $(0, +\infty)$, by essentially the same argument as in the proof of Theorem \ref{mtheorem}. 




Throughout the paper, we denote $\rho^N (t, x) :=  \int_{\mathbb{R}^2} f^N (t, x, v) dv$, and $\textbf{j}^N (t, x) := \int_{\mathbb{R}^2} \hat{v} f^N (t, x , v) dv$. Without loss of generality, we assume that $N$ is large enough, such that the following holds:  
\begin{equation} \label{cond-N}
N \geq 2, \ \text{and} \ 
dist(\supp_x f_0 (x,v), \supp B_{ext,N} (x)) >0 .  
\end{equation}
Later in the paper we will mention several additional requirements on the lower bound of $N$. None of these constraints on $N$ affect our result, since we only care about the scenario when $N \rightarrow +\infty$.

The contents in the paper are arranged as follows. In Section \ref{SectionSetup}, we discuss some details in the set up of the problem. In Section \ref{SectionWL}, we prove some bounds for the particle momentum as well as for the internal electromagnetic field, and obtain some limit for $(f^N, E^N, B^N)$ by extracting subsequences. In this section we also prove the confinement effect carried by the external magnetic field. In Section \ref{ModelCase}, we consider a \emph{model} trajectory ODE system, in which we drop the internal fields, and prove that the external magnetic field $B_{ext,N}$ has a "reflective" effect on charged particles when the internal fields are absent. We illustrate that how this effect resembles the role of a specular boundary condition. This enables us to carry out a perturbative analysis on the trajectory ODEs and explain the "reflective" effect of $B_{ext,N}$ on charged particles when the internal fields come into play. We give the proof of Theorem \ref{mtheorem} in Section \ref{ProofofMainTheorem}. The case when the external magnetic field $B_{ext,N}$ is finite is addressed in Section \ref{SectionFinite}. In Section \ref{SectionTwoSpecies} we analyze the same problem with a plasma that contains both ions and electrons. The appendix is devoted to a tool lemma for the readers' convenience.

This paper is the first article to address the phenomenon of magnetic confinement at the boundary approximating specular reflecting wall from a mathematical point of view. So far we only consider the weak convergence of $(f^N , E^N_1 , E^N_2 , B^N)$ to $(f , E_1 , E_2 , B)$. There are a lot of open directions on this topic, for example, strengthening the weak convergence of $(f^N , E^N_1 , E^N_2 , B^N)$ towards $(f , E_1 , E_2 , B)$ to one in a stronger sense, or investigating the limiting process for higher dimensional settings.

\section{Setup} \label{SectionSetup}

Let $b(x)$ be some $C^1$, piecewise $C^3$, compactly supported function on $(0, +\infty)$ and satisfies
\begin{equation} \label{b-prop}
\begin{split}
&   b' (x) >0  , \ b'' (x) <0 , \text{ and } b''' (x) >0 \text{ on } (0, 1) , \\
& b (x) = 0 \ \text{when} \ x \in (1, +\infty) , \ b (x) \rightarrow -\infty \ \text{as} \ x \rightarrow 0 . \\
\end{split}
\end{equation}
Then we define $B_{ext,N}$ as 
\begin{equation}  \label{BextN-b-rescale}
\begin{split}
& B_{ext, N} (x) := N b(Nx) \text{ for } x \in (0, \frac{1}{2}], \\
& B_{ext, N} (x) := -N b(N(1-x)) \text{ for } x \in [\frac{1}{2}, 1) . \\
\end{split}
\end{equation}
Let $\psi_{ext,N}$ be a magnetic potential for $B_{ext, N} (x)$ defined as
\begin{equation} \label{psiextN-potentialdef}
\psi_{ext,N} (x) :=  \int^x_{1/2}  B_{ext,N} (y) dy .
\end{equation}
Then we have $B_{ext, N} (x) = \partial_x \psi_{ext, N} (x)$, $ \psi_{ext,N} (\frac{1}{2}) = 0  $ 
and $\psi_{ext, N}$ is a piecewise $C^4$, compactly supported function on $\Omega = (0, 1)$ that blows up to $+\infty$ when $x$ approaches $\partial \Omega$. Let 
\begin{equation} \label{Psi-potentialdef}
\Psi (x)  :=  \int^x_{1}  b (y) dy .
\end{equation}
Then $b(x) = \Psi' (x)$, and $\Psi$ is a $C^2$, piecewise $C^4$, compactly supported function $\Psi$ on $(0, +\infty)$ satisfying 
\begin{equation} \label{Psi-prop}
\begin{split}
& \Psi (x) >0,  \ \Psi'' (x) >0 , \ \Psi''' (x) <0    \
\text{and } \Psi'''' (x) >0  \text{ on } (0, 1) ,   \\
& \Psi (x) = 0 \ \text{when} \ x \in (1, +\infty) , \ \Psi (x) \rightarrow +\infty \ \text{as} \ x \rightarrow 0 . \\
\end{split}
\end{equation}
Moreover, $\psi_{ext, N} $ can be viewed as a function obtained by transforming $\Psi$ as follows:
\begin{equation} \label{psiextN-def}
\begin{split}
& \psi_{ext, N} (x) = \Psi(Nx) \text{ for } x \in (0, \frac{1}{2}], \\
& \psi_{ext, N} (x) = \Psi(N(1-x)) \text{ for } x \in [\frac{1}{2}, 1) . \\
\end{split}
\end{equation}




The regularity and monotonicity conditions on $b(x)$ and $\Psi (x)$ as well as the assumptions that $b(x)$ and $\Psi (x)$ are compactly supported are not essential -- they are just set up for technical convenience. (For example, in fact, the main result still holds for the case when $b(x)$ and $\Psi (x)$ only decay to $0$ as $x \rightarrow +\infty$.) Notice that for each $N$, $\psi_{ext, N}$ is piecewise $C^4$, and
\begin{equation}
\psi_{ext, N} (x) \rightarrow +\infty \ \text{as} \ x \rightarrow 0, \ 1 ,
\end{equation}
and
\begin{equation}
\psi_{ext, N} (x) = 0 \ \text{on} \ (\frac{1}{N}, 1- \frac{1}{N} ) .
\end{equation}


For the system \eqref{VlasovBext} -- \eqref{MaxwellBext}, we assume the initial condition 
\begin{equation} 
\begin{split}
& 0 \leq f(0, x, v) = f_0 (x, v) \in C^1_0 (\Omega \times\mathbb{R}^2) \cap L^1  (\Omega \times\mathbb{R}^2)  , \\
& \supp_{x, v} f_0 (x, v) \subset [\epsilon_0, 1-\epsilon_0] \times \{|v| \leq k_0 \}  , \\
& E_2(0, x) = E_{2, 0} (x) \in C^1 , \ B(0, x) = B_{0} (x)  \in C^1 . \\
\end{split}
\end{equation}
Here the constants satisfy $\epsilon_0 \in (0, 1/2)$, $k_0 >0$, and hence $f_0$ is supported away from $\partial \Omega$ with a positive distance to it.

We also need boundary conditions for $E$ and $B$. 
At each boundary point $x=0$ and $x=1$, either $E_2$ or $B$ should be specified, which leaves four possible combinations of boundary conditions for $E_2$ and $B$, see \cite{NS4}. In this paper we take one of the four choices and assume the following boundary conditions for $E_2$ and $B$:
\begin{equation} 
E_2 (t, 0) = E_{2, b} (t)  \in C^1  , \  B (t, 1)  = B_b (t)  \in C^1 .
\end{equation}
The proofs of the main theorem for the other three choices of boundary conditions are similar and we omit them. 

Moreover, we take 
$$   E_1 (0,0) = \lambda $$
as in \cite{NS3}. Here $\lambda$ is a real constant. As in Section 2.1 in \cite{NS3}, we integrate the Vlasov equation \eqref{VlasovBext} or \eqref{VlasovBextspecular} to obtain $\partial_t \rho + \partial_x j_1 =0$. Notice that in both settings (\eqref{VlasovBext} or \eqref{VlasovBextspecular} with the specular boundary condition \eqref{specularBC}) we have $j_1 (t, 0) = j_1 (t, 1)=0$. We integrate $\partial_t \rho + \partial_x j_1 =0$ in $x$ and obtain $\int_\Omega \rho (t, x) dx = \int_\Omega \rho (0, x) dx \equiv \|f_0\|_{L^1 (\Omega \times \mathbb{R}^2)}$. Integrating $\partial_x E_1 = \rho$ and using $ E_1 (0,0) = \lambda $, we deduce $E_1 (t, x) = \int^x_0 \rho (t, y) dy +C(t)$ with $C(0) = \lambda$. From $\partial_t E_1 = -j_1$ and $\partial_t \rho + \partial_x j_1 =0$, we have $C'(t) = -j_1 (t,x) + \int^x_0 \partial_x j_1 (t, y) dy = -j_1 (t, 0) =0$. Therefore $C(t) \equiv \lambda$, which enables us to deduce $ E_1(t, x) = \int^x_0 \int_{\mathbb{R}^2} f (t, y, v) dv dy + \lambda  $. Hence we have the following initial-boundary condition for $E_1$:
\begin{equation}
E_1(0, x) = \int^x_0 \int_{\mathbb{R}^2} f_0 (y, v) dv dy + \lambda  =: E_{1, 0} (x) \in C^1 , \ E_1(t, 0) \equiv \lambda .
\end{equation}
To summarize, we put down the following initial-boundary conditions:
\begin{equation} \label{boundarycondBext}
\begin{split}
& 0 \leq f(0, x, v) = f_0 (x, v) \in C^1_0 (\Omega \times\mathbb{R}^2)  , \\
& \supp_{x, v} f_0 (x, v) \subset [\epsilon_0, 1-\epsilon_0] \times \{|v| \leq k_0 \}  , \\
& E_1(0, x) = \int^x_0 \int_{\mathbb{R}^2} f_0 (y, v) dv dy + \lambda  =:  E_{1, 0} (x) \in C^1 , \  \ E_1(t, 0) \equiv \lambda , \\
& E_2(0, x) = E_{2, 0} (x) \in C^1 , \ B(0, x) = B_{0} (x)  \in C^1 , \\
& E_2 (t, 0) = E_{2, b} (t)  \in C^1  , \  B (t, 1)  = B_b (t)  \in C^1 . \\
\end{split}
\end{equation}
Here the functions $E_{2, b} $, $E_{2, 0}$, $B_b$ and $B_0$ should satisfy
\begin{equation}
E_{2,b} (0) = E_{2, 0} (0), \ B_b (0) = B_0 (1)  
\end{equation}
for the sake of compatibility.


Later in Lemma \ref{NS3-lm3.1}, we will show that there exists a constant $C_v :=  k_0 + C_1 T > 0$ only depending on the initial-boundary data and $T$ (in particular, independent of $N$), such that 
\begin{equation}
\sup_N \sup \{ |v|: f^N (t, x, v) \neq 0 \text{ for some } x \in \Omega \} \leq C_v \text{ for all } t \in [0, T] . 
\end{equation}
We use $D_{C_v}$ to denote the disk on $\mathbb{R}^2$ centered at the origin with radius $C_v$. Therefore 
\begin{equation}
\supp f^N \subset \mathcal{K} 
\end{equation} 
with $\mathcal{K}$ being some compact subset of $[0, T] \times \Omega \times \mathbb{R}^2$.

We introduce the following definition for the weak formulation of the 1.5D RVM with the external magnetic field:

\vskip 0.5cm

\begin{definition}
\label{D:weak-Bext}
\begin{flushleft}
{\it{(Weak solution of the 1.5D RVM with the external magnetic field)}} \\
Let 
\begin{equation*}
\begin{split}
& f^N \geq 0 , \  f^N \in L^1_{loc} ([0, T) \times \Omega \times \mathbb{R}^2)  , \ 
 E^N_1 , E^N_2 , B^N \in L^1_{loc} ([0, T) \times \overline{\Omega}) , \\ 
& \text{and }  \supp f^N \subset \mathcal{K} 
\text{ with } \mathcal{K} \text{ being some compact subset of } [0, T] \times \Omega \times \mathbb{R}^2 .  \\
\end{split}
\end{equation*}
We say that $(f^N, E^N_1, E^N_2, B^N)$ is a weak solution of \eqref{VlasovBext}, \eqref{MaxwellBext} with the initial-boundary conditions \eqref{boundarycondBext} if for any $\alpha (t, x, v) \in C^\infty_c ([0, T) \times \overline{\Omega} \times \mathbb{R}^2)$, and $\varphi_j (t, x) \in C^\infty_c ([0, T) \times \overline{\Omega}  )$, $j = 1, \  2, \ 3, \ 4$ satisfying
$$   \varphi_2 (t, 1) = \varphi_3 (t, 0) =  \varphi_4 (t, 1) =0 ,  $$
the following holds:
\end{flushleft}
\begin{equation} \label{weakformVlasovN}
\begin{split}
&  \int^1_0 \int_{\mathbb{R}^2} \int^T_0 f^N \cdot \Big\{  \partial_t \alpha + \hat{v}_1   \partial_x \alpha  + ( E^N_1 + \hat{v}_2 B^N + \hat{v}_2 B_{ext, N} ) \partial_{v_1} \alpha  \\
& + ( E^N_2 - \hat{v}_1 B^N - \hat{v}_1 B_{ext, N}  )   \partial_{v_2}  \alpha  \Big\} dt dv dx  + \int^1_0 \int_{\mathbb{R}^2} f_0 (x, v) \alpha (0, x, v) dv dx  = 0 , \\
\end{split} 
\end{equation}
\begin{equation} \label{weakformMaxwellN1}
\begin{split}
& -\int^1_0  \int^T_0 E^N_1 \partial_t \varphi_1 dt dx
  - \int^1_0   E_{1, 0} (x) \varphi_1 (0, x) dx  
+  \int^1_0 \int^T_0 \int_{\mathbb{R}^2} \hat{v}_1 f^N  \varphi_1 dv dt dx   =0 ,  \\
\end{split}
\end{equation}
\begin{equation} \label{weakformMaxwellN2}
\begin{split}
& -\int^1_0  \int^T_0 E^N_1 \partial_x \varphi_2 dt dx  
 - \lambda \int^T_0   \varphi_2 (t, 0) dt   
 - \int^1_0 \int^T_0 \int_{\mathbb{R}^2} f^N \varphi_2 dv dt dx =0 , \\
\end{split}  
\end{equation}
\begin{equation} \label{weakformMaxwellN3}
\begin{split}
& -\int^1_0 \int^T_0  E^N_2 \partial_t \varphi_3 dt dx - \int^1_0   E_{2, 0} (x) \varphi_3 (0, x) dx    \\
& - \int^1_0 \int^T_0 B^N \partial_x \varphi_3 dt dx  + \int^T_0   B_b (t) \varphi_3 (t, 1) dt
  + \int^1_0 \int^T_0 \int_{\mathbb{R}^2} \hat{v}_2 f^N \varphi_3 dv  dt dx  =0 , \\
\end{split}
\end{equation}
and
\begin{equation} \label{weakformMaxwellN4}
\begin{split}
& -\int^1_0 \int^T_0 B^N \partial_t \varphi_4 dt dx  - \int^1_0  B_0 (x) \varphi_4 (0, x) dx    \\
& - \int^1_0 \int^T_0  E^N_2 \partial_x \varphi_4 dt dx - \int^T_0  E_{2, b} (t) \varphi_4 (t, 0) dt  
=0 .  \\ 
\end{split}
\end{equation}
\end{definition}

\vskip 0.5cm

We also introduce the following definition for the weak formulation of the 1.5D RVM for the case without the external magnetic field but with the specular boundary condition:

\vskip 0.5cm

\begin{definition}
\label{D:weak}
\begin{flushleft}
{\it{(Weak solution of the 1.5D RVM with the specular boundary condition)}} \\
Let 
\begin{equation*}
\begin{split}
& f \geq 0 , \  f \in L^1_{loc} ([0, T) \times \overline{\Omega} \times \mathbb{R}^2) , \  
 E_1 , E_2 , B \in L^1_{loc} ([0, T) \times \overline{\Omega}) . \\ 
\end{split}
\end{equation*}
We say that $(f, E_1, E_2, B)$ is a weak solution of \eqref{VlasovBextspecular}, \eqref{MaxwellBextspecular} with the initial-boundary conditions \eqref{boundarycondBext} together with the specular boundary condition \eqref{specularBC} if for any $\alpha (t, x, v) \in C^\infty_c ([0, T) \times \overline{\Omega} \times \mathbb{R}^2)$ and $\varphi_j (t, x) \in C^\infty_c ([0, T) \times \overline{\Omega}  )$, $j = 1, \  2, \ 3, \ 4$ satisfying
$$\alpha(t,0,v_1,v_2) = \alpha(t,0,-v_1,v_2) , $$
$$\alpha(t,1,v_1,v_2) = \alpha(t,1,-v_1,v_2) , $$
$$  \varphi_2 (t, 1) = \varphi_3 (t, 0) = \varphi_4 (t, 1) =0 ,  $$
the following holds:
\end{flushleft}
\begin{equation} \label{E:weakformVlasov}
\begin{split}
&  \int_0^1 \int_{\mathbb{R}^2} \int_0^T f \cdot \Big\{ \partial_t \alpha + \hat{v}_1 \partial_x \alpha + ( E_1 + \hat{v}_2 B ) \partial_{v_1} \alpha  + ( E_2 - \hat{v}_1 B ) \partial_{v_2} \alpha \Big\} dt dv dx  \\
& + \int_0^1 \int_{\mathbb{R}^2} f_0 (x, v) \alpha (0, x, v)  dv dx  = 0 ,\\
\end{split}
\end{equation}
\begin{equation} \label{E:weakformMaxwell1}
\begin{split}
& -\int^1_0  \int^T_0 E_1 \partial_t  \varphi_1 dt dx   
- \int^1_0   E_{1, 0} (x) \varphi_1 (0, x) dx  
+ \int^1_0 \int^T_0 \int_{\mathbb{R}^2} \hat{v}_1 f  \varphi_1 dv dt dx  =0 ,  \\
\end{split}
\end{equation}
\begin{equation} \label{E:weakformMaxwell2}
\begin{split}
& -\int^1_0  \int^T_0 E_1 \partial_x \varphi_2 dt dx  
 - \lambda \int^T_0  \varphi_2 (t, 0) dt   
- \int^1_0 \int^T_0 \int_{\mathbb{R}^2} f \varphi_2 dv dt dx =0 , \\
\end{split}  
\end{equation}
\begin{equation} \label{E:weakformMaxwell3}
\begin{split}
& -\int^1_0 \int^T_0  E_2 \partial_t \varphi_3 dt dx - \int^1_0   E_{2, 0} (x) \varphi_3 (0, x) dx    \\
& - \int^1_0 \int^T_0 B \partial_x \varphi_3 dt dx + \int^T_0   B_b (t) \varphi_3 (t, 1) dt 
  + \int^1_0 \int^T_0 \int_{\mathbb{R}^2} \hat{v}_2 f \varphi_3 dv  dt dx  =0 , \\
\end{split}
\end{equation}
and
\begin{equation} \label{E:weakformMaxwell4}
\begin{split}
& -\int^1_0 \int^T_0 B \partial_t \varphi_4 dt dx  - \int^1_0  B_0 (x) \varphi_4 (0, x) dx    \\
& - \int^1_0 \int^T_0  E_2 \partial_x \varphi_4 dt dx - \int^T_0  E_{2, b} (t) \varphi_4 (t, 0) dt  
=0 .  \\ 
\end{split}
\end{equation}
Note that the class where the test functions $\alpha$ belong includes those that are compactly supported in $\Omega$.  
\end{definition}

Notice that the only difference between Definition \ref{D:weak-Bext} and Definition \ref{D:weak} lies in the weak form of the Vlasov equation. We give some explanation for \eqref{E:weakformVlasov} here. Assume $(E_1, E_2, B) \in C^1$. It is obvious that a $C^1$ solution $f$ of the Vlasov equation \eqref{VlasovBextspecular} and the initial-boundary conditions \eqref{boundarycondBext} together with the specular boundary condition \eqref{specularBC} satisfies Definition \ref{D:weak} by noticing that $\int_{\mathbb{R}^2} \int^T_0 \hat{v}_1 f (t, 0, v) \alpha (t, 0, v) dt dv = \int_{\mathbb{R}^2} \int^T_0 \hat{v}_1 f (t, 1, v) \alpha (t, 1, v) dt dv =0$ holds for $\alpha$ satisfying $\alpha(t,0,v_1,v_2) = \alpha(t,0,-v_1,v_2) $ and $\alpha(t,1,v_1,v_2) = \alpha(t,1,-v_1,v_2) $ if $f$ satisfies \eqref{specularBC}. Conversely, let $f$ be a $C^1$ function that satisfies Definition \ref{D:weak}. The usual weak form of the Vlasov equation \eqref{VlasovBextspecular} is
\begin{equation}  \label{E:weakformVlasov-usual}
\begin{split}
&  \int^1_0 \int_{\mathbb{R}^2} \int^T_0 \Big\{ f \partial_t \alpha + \hat{v}_1 f \partial_x \alpha + ( E_1 + \hat{v}_2 B )f \partial_{v_1} \alpha  + ( E_2 - \hat{v}_1 B   ) f \partial_{v_2}  \alpha  \Big\} dt dv dx  \\
& + \int_{\mathbb{R}^2} \int^T_0 \hat{v}_1 f (t, 0, v) \alpha (t, 0, v) dt dv - \int_{\mathbb{R}^2} \int^T_0 \hat{v}_1 f (t, 1, v) \alpha (t, 1, v) dt dv  \\
& + \int^1_0 \int_{\mathbb{R}^2} f (0, x, v) \alpha (0, x, v) dx dv  = 0  , \ \forall \alpha \in C^\infty_c ([0, T) \times \overline{\Omega} \times \mathbb{R}^2 )  . \\
\end{split}
\end{equation}
For $x \notin \partial \Omega$, we take $\alpha$ that satisfies $\supp_x \alpha \subset (0, 1)$ and it is easy to see that \eqref{E:weakformVlasov} implies that  \eqref{VlasovBextspecular} holds for all $(t, x, v_1, v_2) \in [0, T) \times \Omega \times \mathbb{R}^2$ (with $(f, E_1, E_2, B) \in C^1$).
For the case $x \in \partial \Omega$, it suffices to consider the boundary $x=0$. For any $t_0 \in [0, T)$, $v_{1,0} \neq 0$, $v_{2,0} \in \mathbb{R}$, there exists a sequence, such that $\alpha_j (t_0, 0, v_{1, 0}, v_{2, 0}) = \alpha_j (t_0, 0, -v_{1, 0}, v_{2, 0}) = j$, and as $j \rightarrow \infty$, $\alpha_j (t, x, v_1, v_2)$ converges to  
\begin{equation}
\label{E:key-test-model}
\alpha_0 (t,x,v_1,v_2) = [\delta_{v_{1,0}}(v_1) + \delta_{-v_{1,0}} (v_1)] \delta_{t_0}(t)\delta_{v_{2,0}}(v_2) \delta_0 (x)
\end{equation}
in the sense of distribution. Notice that $\alpha_0$ meets the even condition with respect to $v_1$. Plugging $\alpha_j$ into \eqref{E:weakformVlasov} and letting $j \rightarrow \infty$ gives \eqref{VlasovBextspecular} at $(t_0, 0, v_{1,0}, v_{2,0})$ as well as
$$\hat v_{1,0} f(t_0,0,v_{1,0},v_{2,0}) - \hat v_{1,0} f(t_0,0,-v_{1,0},v_{2,0}) =0 . $$
Cancelling out $\hat v_{1,0}$ in the last equality gives the classical specular boundary condition \eqref{specularBC}. 

\vskip 0.5cm

\section{Bounds for Particle Momentum and Electromagnetic Field} \label{SectionWL}

We first prove some bounds for the particle momentum and the electromagnetic field, and obtain some limit object for $(f^N, E^N, B^N)$ by extracting subsequences.

From Section 2 in \cite{NS3}, we have the following lemma giving a uniform $L^\infty$ bound on $(E^N_1, E^N_2, B^N)$:

\begin{lemma}   \label{NS3-Section2-lm}
The sequence $\{(E^N_1, E^N_2, B^N)\}$ satisfies
\begin{equation} \label{NS3-Section2}
\|(E^N_1, E^N_2, B^N)\|_{L^\infty ([0, T] \times \Omega)} \leq C_1 
\end{equation}
where $C_1$ is a positive constant defined as
\begin{equation} \label{NS3-Section2-C1}
\begin{split}
C_1 
& := \|f_0 \|_{L^1 (\Omega \times \mathbb{R}^2)} + \lambda + \|E_{2, 0}\|_{L^\infty (\Omega)} + \|E_{2, b}\|_{L^\infty ([0, T] \times \partial \Omega)} + \|B_0\|_{L^\infty (\Omega)} + \| B_b\|_{L^\infty ([0, T] \times \partial \Omega)} \\
& \quad + \frac{1}{4} [ (\|f_0\|_{L^1 (\Omega \times \mathbb{R}^2)} + \lambda )^2 + \|E_{2, 0}\|^2_{L^\infty (\Omega)} + \|B_0 \|^2_{L^\infty (\Omega)} + 4T \|E_{2, b} B_b \|_{L^\infty ([0, T] \times \partial \Omega)} ] \\
& \quad + \frac{1}{2} \| \la v \ra f_0 \|_{L^1 (\Omega \times \mathbb{R}^2)} . \\ 
\end{split}
\end{equation}
\end{lemma}
\begin{proof}
The proof is given in Corollary 2.4 in \cite{NS3} so we omit it here.
\end{proof}

\vskip 0.3cm


The next lemma is introduced to describe the bound for the particle momentum as well as the relation between the confining potential and the particle trajectory.


\begin{lemma} \label{NS3-lm3.1}
Suppose $\supp_{x, v} f_0 (x, v) \subset [\epsilon_0, 1-\epsilon_0] \times \{|v| \leq k_0 \} $. 
Denote $P_N (t) :=  \sup \{ |v|: f^N (t, x, v) \neq 0 \text{ for some } x \in \Omega \}$. We have: \\
1)
\begin{equation}  \label{NS3-lm3.1-PN}
P_N (t) \leq C_v :=  k_0 + C_1 T , \text{ for all } t \in [0, T] . 
\end{equation}
Hence the support of $f^N$ in $v$ is contained in the disk $\overline{D}_{C_v}$. \\
2) When $N \geq \epsilon_0^{-1}$, 
\begin{equation} \label{NS3-lm3.1-psiextN}
\|\psi_{ext, N} \|_{L^\infty (\supp_x f^N )} \leq C_2 .
\end{equation}
Here $ C_2  :=   2 k_0 + 2 C_1 T + 2 C_1 $, and $C_1$ is given in \eqref{NS3-Section2-C1}.  
\end{lemma}
{\it{Remark.} The inequality \eqref{NS3-lm3.1-psiextN} tells us that the support of $f^N$ in $x$ stays away from the boundary $\partial \Omega$ with a positive distance, i.e. $dist ( \supp_x f^N, \partial \Omega) > 0$ on $[0, T]$. }
\begin{proof}
The lemma follows from Lemma 3.1 and 3.4 in \cite{NS3}. We provide the proof here for completeness.

The ODE for the particle trajectory is
\begin{equation}
\left\{
\begin{aligned}
&\dot X_N = \hat V_{1,N} \\
&\dot V_{1,N} =  E^N_1 (t, X_N) + \hat V_{2,N} B^N (t, X_N) + \hat V_{2,N} B_{ext, N} (X_N) \\
&\dot V_{2,N} =  E^N_2 (t, X_N) - \hat V_{1,N} B^N (t, X_N) - \hat V_{1,N} B_{ext, N} (X_N)
\end{aligned}
\right.
\end{equation}
with initial data $X_N(0) = x$, $V_{1,N} (0) = v_1$, $V_{2,N} (0) = v_2$. We compute, using the ODE above, (here $\dot{F}$ means $\frac{\partial F}{\partial s}$ for any function $F$)
\begin{equation*}
\frac{d}{ds} |V_N|^2 = 2(V_{1,N} \dot{V}_{1,N} + V_{2,N} \dot{V}_{2,N}) = 2V_N \cdot E^N (s, X(s)) .
\end{equation*}
Hence by \eqref{NS3-Section2}, we obtain
\begin{equation*}
|V_N(s)|^2 \leq |v|^2 + 2C_1 \int^s_0 |V_N(\tau)| d\tau , \ \text{for} \ s \in [0, T] . 
\end{equation*}
By the quadratic Gronwall lemma, we have
\begin{equation*}
|V_N(s)| \leq |v| + C_1 T \ \text{for} \ s \in [0, T] . 
\end{equation*} 
By the definition of $P_N (t)$, we have
\begin{equation*}
P_N (t) \leq k_0 + C_1 T , \text{ for all } t \in [0, T] . 
\end{equation*}
1) is proved.

Next, let $\psi^N (\tau, y) = \int^y_{1/2} B^N(\tau, z) dz$. 
We define
\begin{equation*}
p(\tau, y, w)  :=  w_2 +\psi^N (\tau,y) + \psi_{ext,N} (y) ,
\end{equation*}
where $w = (w_1, w_2) \in \mathbb{R}^2$. Differentiating $p(\tau, y, w)$ along the characteristics, we obtain
\begin{equation*}
\begin{split}
& \quad \frac{d}{ds} p(s, X_N (s), V_N (s)) \\
& = \dot{V}_{2,N} (s) + \partial_t \psi^N (s, X_N (s)) + \dot{X}_N \partial_x \psi^N (s, X_N (s)) + \dot{X}_N \partial_x \psi_{ext,N} ( X_N (s)) \\
& = E^N_2 (s, X_N(s)) - \hat{V}_{1,N} (s) [B^N (s,X_N(s)) + B_{ext,N} ( X_N(s))] \\
& \quad + \partial_t \psi^N (s, X_N (s)) + \hat{V}_{1,N} (s) [B^N (s,X_N(s)) + B_{ext,N} ( X_N(s))] \\
& = E^N_2 (s, X_N(s)) + \partial_t \psi^N (  X_N (s)) \\
& = E^N_2 (s, \frac{1}{2}) . \\
\end{split}
\end{equation*}
Here we used the fact that $\partial_t B^N = -\partial_x E^N_2$. Integrating yields
\begin{equation*}
V_{2,N}(s) + \psi^N(s, X_N(s))+ \psi_{ext,N}(X_N(s)) = v_2 + \psi^N (0,x) + \psi_{ext,N} (x) + \int^s_0 E^N_2 (\tau, \frac{1}{2}) d\tau ,
\end{equation*}
and hence
\begin{equation*}
|\psi_{ext,N}(X_N(s))| \leq 
|V_{2,N}(s)| + |\psi^N(s, X_N(s))| + | v_2| + |\psi^N (0,x)| + |\psi_{ext,N} (x)| + \int^s_0 |E^N_2 (\tau, \frac{1}{2})| d\tau .
\end{equation*}
Combining this with \eqref{NS3-Section2} and \eqref{NS3-lm3.1-PN}, we have
\begin{equation*}
|\psi_{ext,N}(X_N(s))| \leq \|\psi_{ext, N} \|_{L^\infty ([\epsilon_0, 1-\epsilon_0])} + C_2 , \text{ for all } s \in [0, T] .
\end{equation*}
Notice that when $N \geq \epsilon_0^{-1}$, $ \|\psi_{ext, N} \|_{L^\infty ([\epsilon_0, 1-\epsilon_0])} =0$. Hence 
\begin{equation*}
|\psi_{ext,N}(X_N(s))| \leq  C_2 , \text{ for all } s \in [0, T] .
\end{equation*}
This inequality holds for all the trajectories. Therefore we conclude
\begin{equation*}
\|\psi_{ext, N} \|_{L^\infty (\supp_x f^N )} \leq C_2 .
\end{equation*}
The proof of 2) is complete.
\end{proof}
 
\begin{corollary} \label{NS3-lm3.1-cor}
There exists $y_0 >0$ (depends on $f_0$) independent of $N$ and small enough such that $\supp_x f^N (t) \subset (N^{-1} y_0, 1-N^{-1} y_0)$ for all $t \in [0, T]$. For any $x \in \supp_x f^N (t)$ ($t \in [0, T]$), we have $|\Psi (Nx) |\leq |\Psi (y_0)|$, $|\Psi' (Nx) |\leq |\Psi' (y_0)|$, $|\Psi'' (Nx) |\leq |\Psi'' (y_0)|$, $|\Psi''' (Nx) |\leq |\Psi''' (y_0)|$. 
\end{corollary}

\begin{proof}
Without loss of generality, we can assume $C_2 \geq \Psi (1)$. By \eqref{NS3-lm3.1-psiextN} and the monotonicity of $\Psi$ (see \eqref{Psi-prop}), we can define $y_0 :=  \Psi^{-1} (C_2) $ (where $\Psi^{-1}$ means the inverse function of $\Psi$ defined on $[\Psi(1), +\infty)$). Then $\supp_x f^N \subset (N^{-1} y_0, 1-N^{-1} y_0)$ by \eqref{psiextN-def}. The rest of corollary follows from the monotonicity of $\Psi$ and its derivatives given in \eqref{Psi-prop}.
\end{proof}



Now we are ready to introduce the following estimate for the derivatives of $E_2^N$ and $B^N$, whose proof is similar to the one for Lemma 4.1 in \cite{NS3}:

\begin{lemma} \label{NS3-lm4.1}
There exists a constant $C_T$ which only depends on the initial-boundary data and $T$, such that $\|\partial_{x} E_1^N \|_{L^\infty ([0, T] \times \Omega)} \leq C_T$, $\|\partial_{x} E_2^N \|_{L^\infty ([0, T] \times \Omega)} \leq C_T$, $\|\partial_{x} B^N \|_{L^\infty ([0, T] \times \Omega)} \leq C_T$. In particular, $C_T$ is independent of $N$.  
\end{lemma}

\begin{proof}
The bounds for $\|\partial_{x} E_1^N \|_{L^\infty ([0, T] \times \Omega)} $ follows easily from the equation $\partial_x E^N_1 = \rho^N$:
\begin{equation} \label{lm-4.1-eq0}
\|\partial_{x} E_1^N \|_{L^\infty ([0, T] \times \Omega)} \leq \int_\Omega \int_{\mathbb{R}^2} f^N (t, x, v) dv dx  \leq  \int_\Omega \pi C_v^2 \|f_0 \|_{L^\infty (\Omega \times \mathbb{R}^2)}  dx  \leq \pi C_v^2 \|f_0 \|_{L^\infty (\Omega \times \mathbb{R}^2)} .
\end{equation}
It suffices to prove the bounds for $\|\partial_{x} E_2^N \|_{L^\infty ([0, T] \times \Omega)}$ and $\|\partial_{x} B^N \|_{L^\infty ([0, T] \times \Omega)}$. The proof is modified from the one for Lemma 4.1 in \cite{NS3}. It suffices to derive the $L^\infty ([0, T])$ estimate on $[0, T] \times (0, 1/2]$ since the case $x > 1/2$ is similar (the only change being that in below we express $\partial_x = \frac{S - T_- }{1+\hat{v}_1} $, where $T_- = \partial_t - \partial_x$, $S = \partial_t + \hat{v}_1 \partial_x$). 

Let $y_0$ be as defined in Corollary \ref{NS3-lm3.1-cor} and $\theta_0 := N^{-1} y_0$. Let 
$$k^{N,\pm} (t, x) := (E^N_2 \pm B^N)(t, x) . $$ 
It suffices to show $ \| \partial_x k^{N, \pm} (t, x) \|_{L^\infty ([0, T] \times (0, 1/2])} \leq C_T$. We only need to deal with $\partial_x  k^{N, +}$ since the bound for $\partial_x k^{N, -}$ is obtained in a similar manner.

By the argument leading to Lemma 2.1 in \cite{NS3}, we have, for $(t, x) \in [0, T] \times (0, 1/2]$,
\begin{equation} \label{lm-4.1-eq1}
\begin{split}
& k^{N,+} (t, x) = \frac{1}{2} A^+ (x-t) - \int^t_{t^+(x)} j^N_2 (\tau, x-t+\tau) d\tau , \\
& k^{N,-} (t, x) = \frac{1}{2} A^- (x-t) - \int^t_{t^-(x)} j^N_2 (\tau, x+t-\tau) d\tau . \\
\end{split}
\end{equation}
Here $A^\pm$ are given explicitly in terms of the initial-boundary data, and
\begin{equation*}
t^+ (x) :=  (t-x) \mathbf{1}_{t> x} , \ t^- (x) :=  (t-1+x) \mathbf{1}_{t> 1-x} , 
\end{equation*}
as defined in \cite{NS3}.
Differentiating the $k^{N,+}$ identity in \eqref{lm-4.1-eq1} with respect to $x$, we obtain
\begin{equation} \label{lm-4.1-eq2}
\partial_x k^{N, +} (t, x) = M^N (t, x) - \int^t_{t^+(x)} \partial_x j^N_2 (\tau, x-t+\tau) d\tau = M^N (t, x) - \int^t_{t^+(x)}  \int_{\mathbb{R}^2} \hat{v}_2 \partial_x f^N (\tau, x-t+\tau, v) dv d\tau .
\end{equation}
Here
\begin{equation} \label{lm-4.1-eq3}
M^N (t, x)  :=  \frac{1}{2} (A^+)' (x-t) + j^N_2 (t^+ (x), x-t+ t^+ (x)) (t^+)' (x) . 
\end{equation}

We have, due to Corollary 3.5 in \cite{NS3}, 
\begin{equation}
\| M^N \|_{L^\infty ([0, T] \times \Omega)} \leq C_M , 
\end{equation}
where the constant $C_M$ only depends on the initial-boundary data and $T$.


We use the splitting method of Glassey and Strauss (see \cite{GlasseyS1} and \cite{GlasseySchaeffer1}) to express the operator $\partial_x$. Denote
\begin{equation}
T_+  :=  \partial_t + \partial_x , \ S  := \partial_t + \hat{v}_1 \partial_x . 
\end{equation}
Then
\begin{equation}
\partial_x = \frac{T_+ - S}{1-\hat{v}_1} . 
\end{equation}
Denote
\begin{equation*}
K^N :=  ( E^N_1 + \hat{v}_2 B^N + \hat{v}_2 B_{ext, N} ,  E^N_2 - \hat{v}_1 B^N - \hat{v}_1 B_{ext, N} ) .
\end{equation*}
The Vlasov equation can be written as
\begin{equation}
Sf^N + \nabla_v \cdot (K^N f^N) =0 . 
\end{equation}


Using \eqref{lm-4.1-eq2}, the Vlasov equation as well as integration by parts, we obtain
\begin{equation}
\begin{split}
\partial_x k^{N, +} (t, x)
& =  M^N (t, x) - \int^t_{t^+ (x)} \frac{d}{d \tau} \int_{\mathbb{R}^2} \frac{\hat{v}_2}{1-\hat{v}_1} f^N (\tau, x-t+\tau, v) dv d \tau \\
& \quad - \int^t_{t^+(x)} \int_{\mathbb{R}^2} \frac{\hat{v}_2}{1-\hat{v}_1} \nabla_v \cdot (K^N f^N)  (\tau, x-t + \tau, v) dv d \tau    \\
& = M^N (t, x) - \int_{\mathbb{R}^2} \frac{\hat{v}_2}{1-\hat{v}_1} f^N (t, x, v) dv - \int_{\mathbb{R}^2} \frac{\hat{v}_2}{1-\hat{v}_1} f^N (t^+(x), x-t + t^+ (x), v) dv  \\
& + \int^t_{t^+(x)} \int_{\mathbb{R}^2} \nabla_v (\frac{\hat{v}_2}{1-\hat{v}_1} ) \cdot (K^N f^N) (\tau, x-t+\tau, v) dv d\tau . \\
\end{split}
\end{equation}
We know that the support of $f^N$ in $v$ is contained in the disk $\overline{D}_{C_v}$, where $C_v :=  k_0 + C_1 T$ (see Lemma \ref{NS3-lm3.1}). Using $\|f^N \|_{L^\infty} \leq \|f_0 \|_{L^\infty}$, we compute
\begin{equation}
\begin{split}
& \quad \| \partial_x k^{N, +} (t, x) \|_{L^\infty ([0, T] \times (0, 1/2])} \\
& \leq C_M + 2 \pi C_v^2 \| f_0 \|_{L^\infty (\Omega \times \mathbb{R}^2)} \|\frac{\hat{v}_2}{1-\hat{v}_1}\|_{L^\infty (D_{C_v})}   \\
& \quad + \|\nabla_v (\frac{\hat{v}_2}{1-\hat{v}_1}) \|_{L^\infty (D_{C_v})}  \int^x_{x-t+t^+(x)} \int_{D_{C_v}} |K^N f^N| (y-x+t,y, v) dv dy  . \\
\end{split}
\end{equation}
Recall that $K^N  = ( E^N_1 + \hat{v}_2 B^N + \hat{v}_2 B_{ext, N} ,  E^N_2 - \hat{v}_1 B^N - \hat{v}_1 B_{ext, N} )$, and that $\| ( E^N_1 + \hat{v}_2 B^N   ,  E^N_2 - \hat{v}_1 B^N )\|_{L^\infty} \leq 2C_1$. Moreover, we notice that the integrals in $x$ are actually carried out on the interval $(\theta_0, 1-\theta_0)= (N^{-1} y_0, 1-N^{-1} y_0)$. Recall that $B_{ext, N} (x) \leq 0 $ on $(0,  \frac{1}{2}]$, and $B_{ext, N} (x) \geq 0 $ on $[ \frac{1}{2}, 1)$. Combining together all these information together with $\|f^N \|_{L^\infty} \leq \|f_0 \|_{L^\infty}$, we have
\begin{equation}
\begin{split}
& \quad \| \partial_x k^{N, +} (t, x) \|_{L^\infty ([0, T] \times (0, 1/2])} \\
& \leq C_M + 2 \pi C_v^2 \| f_0 \|_{L^\infty (\Omega \times \mathbb{R}^2)} \|\frac{\hat{v}_2}{1-\hat{v}_1}\|_{L^\infty (D_{C_v})} \\
& \quad +   \|\nabla_v (\frac{\hat{v}_2}{1-\hat{v}_1}) \|_{L^\infty (D_{C_v})}    \big\{ 4 \pi C_v^2 C_1 \|f_0\|_{L^\infty (\Omega \times \mathbb{R}^2)} \\
& \quad + \int^{1/2}_{\theta_0} \int_{D_{C_v}} (- B_{ext, N} f^N) (y-x+t,y, v) dv dy  + \int^{1-\theta_0}_{1/2} \int_{D_{C_v}} ( B_{ext, N} f^N) (y-x+t,y, v) dv dy \big\} . \\
\end{split}
\end{equation}
By direct computation we have $\|\frac{\hat{v}_2}{1-\hat{v}_1}\|_{L^\infty (D_{C_v})} \leq 2C_v + 3C_v^2 $,  $\|\nabla_v (\frac{\hat{v}_2}{1-\hat{v}_1} )\|_{L^\infty (D_{C_v})} \leq 2+ 4C_v + 2C_v^2+3C_v^3 $. Moreover, we estimate the terms involving $B_{ext, N}$ using Lemma \ref{NS3-lm3.1} and $\|f^N \|_{L^\infty} \leq \|f_0 \|_{L^\infty}$ (noticing $\psi_{ext, N} (\frac{1}{2}) =0$):  
\begin{equation*}
\begin{split}
 \int^{1/2}_{\theta_0} \int_{D_{C_v}} (-B_{ext, N} f^N) (y-x+t,y, v) dv dy   
& \leq    \pi C_v^2 \|f_0 \|_{L^\infty (\Omega \times \mathbb{R}^2)} [ \psi_{ext, N} (\theta_0) -0] \\
& \leq  \pi C_v^2 |\Psi (y_0)| \|f_0 \|_{L^\infty (\Omega \times \mathbb{R}^2)} . \\
\end{split}
\end{equation*}
and similarly 
\begin{equation*}
\int^{1-\theta_0}_{1/2} \int_{D_{C_v}} (B_{ext, N} f^N) (y-x+t,y, v) dv dy \leq \pi C_v^2 |\Psi (y_0)| \|f_0 \|_{L^\infty (\Omega \times \mathbb{R}^2)} .
\end{equation*}
Plugging the estimates above, we arrive at
\begin{equation} \label{lm-4.1-eq4}
\begin{split}
& \quad \| \partial_x k^{N, +} (t, x) \|_{L^\infty ([0, T] \times (0, 1/2])} \\
& \leq C_M + 2 \pi C_v^2 \| f_0 \|_{L^\infty (\Omega \times \mathbb{R}^2)} (2C_v + 3C_v^2) \\
& \quad + (4 \pi C_v^2 C_1 + 2 \pi C_v^2 |\Psi (y_0)| ) (2+ 4C_v + 2C_v^2+3C_v^3) \| f_0 \|_{L^\infty (\Omega \times \mathbb{R}^2)} \\
& \leq C_T . \\
\end{split}
\end{equation}
where 
\begin{equation}
\begin{split}
C_T 
& := C_M + 2 \pi C_v^2 \| f_0 \|_{L^\infty (\Omega \times \mathbb{R}^2)} (1+ 2C_v + 3C_v^2) \\
& \quad + (4 \pi C_v^2 C_1 + 2 \pi C_v^2 |\Psi (y_0)| ) (2+ 4C_v + 2C_v^2+3C_v^3) \| f_0 \|_{L^\infty (\Omega \times \mathbb{R}^2)} 
\end{split}
\end{equation}
is a positive constant which only depends on the initial-boundary data and $T$, according to Lemma \ref{NS3-lm3.1} and Corollary \ref{NS3-lm3.1-cor}. In particular, $C_T$ is independent of $N$. The $t$-derivatives for the fields then follow from the Maxwell equations. Combining together all the estimates above, we complete the proof of the lemma.
\end{proof}

Combining Lemma \ref{NS3-Section2-lm} and Lemma \ref{NS3-lm4.1}, we obtain

\begin{lemma} \label{lm3}
For any $T>0$, there exists a constant $C >0$ (which only depends on the initial-boundary data and $T$, in particular, independent of $N$), such that for all $N$ large enough such that \eqref{cond-N} holds, 
\begin{equation}
 \|(E^N_1, E^N_2, B^N) \|_{C^1 ([0, T] \times \Omega) } \leq C ,  
\end{equation}
and by the same argument, there exists a constant $C' >0$ (which only depends on the initial-boundary data and $T$, in particular, independent of $N$), such that for all $N$ large enough such that \eqref{cond-N} holds, 
\begin{equation}
 \|(E^N_1, E^N_2, B^N) \|_{C^1 ([0, 2T] \times \Omega) } \leq C' .  
\end{equation}
By Arzela-Ascoli Theorem, there exists a subsequence of $(E^N_1, E^N_2, B^N)$ (still indexed by $N$) that converges strongly in $C^0 ( [0,T] \times \Omega )$.  
\end{lemma}

\begin{proof}
The assertion directly follows from \eqref{NS3-Section2} and Lemma \ref{NS3-lm4.1}.
\end{proof}

On the other hand, for the sequence $\{f^N \}$ we have

\begin{lemma} \label{lm1}
The family $\{ f^N (t, x, v)\}$ is relatively compact in $weak^*-L^\infty (\mathbb{R}_+ \times \Omega \times \mathbb{R}^2)$. Therefore upon extracting subsequence, we have a limit $f$ of $\{f^N\}$ in $weak^*-L^\infty (\mathbb{R}_+ \times \Omega \times \mathbb{R}^2)$. 
\end{lemma}


\begin{proof}
We have $\| f^N \|_{L^\infty (\mathbb{R}_+ \times \Omega \times \mathbb{R}^2) } = \| f_0 \|_{L^\infty  (\mathbb{R}_+ \times \Omega \times \mathbb{R}^2)} $ by the property of the transport equation. Hence $\{ f^N (t, x, v)\}$ is relatively compact in $weak^*-L^\infty (\mathbb{R}_+ \times \Omega \times \mathbb{R}^2)$.
\end{proof}

Combining together Lemma \ref{lm1} and Lemma \ref{lm3}, we obtain

\begin{lemma} \label{lm4}
For each $N$, we consider a $C^1$ solution $(f^N, E^N_1, E^N_2, B^N)$ on $[0, T]$ to \eqref{VlasovBext}, \eqref{MaxwellBext}, with the initial-boundary conditions \eqref{boundarycondBext}. There exists a subsequence of $(f^N , E^N_1 , E^N_2 , B^N)$, such that $f^N$ converges to some $f$ in $weak^*-L^\infty ([0, T] \times \Omega \times \mathbb{R}^2)$, and $(E^N_1 , E^N_2 , B^N)$ converges to some $(E_1, E_2, B)$ strongly in $(C^0  ([0, T] \times \Omega ))^3$.
\end{lemma}

In the rest of the paper, the notations $F = O(1/N)$, $F = O(1/(N \epsilon))$, and $F \lesssim 1/N$, $F \lesssim 1/(N\epsilon)$ mean that $\|F\|_{L^\infty} \leq C_F/N$ for some constant $C_F$ only depending on $\Psi$, $f_0$, $\|(E^N_1, E^N_2, B^N)\|_{C^1_{t, x} ([0, 2T] \times \Omega)}$, $T$ and possibly the test functions selected in the weak formulation of the RVM system (see Definition \ref{D:weak-Bext} and \ref{D:weak}). Notice that $\|(E^N_1, E^N_2, B^N)\|_{C^1_{t, x} ([0, 2T] \times \Omega)}$ are bounded by constants that only depends on the the initial-boundary data, $\Psi$ and $T$. Therefore actually the constant $C_F$ only depends on the initial-boundary data, $\Psi$, $T$ and the test functions involved in the weak formulation of the RVM system. The notations $F \gtrsim 1/N$, $F \gtrsim 1/(N \epsilon)$, $F = O(1)$, $F \lesssim 1$, $F \lesssim \epsilon$, $F \gtrsim 1$, $F \lesssim 1/N^2$, $F \lesssim 1/(N \epsilon^2)$, $F \lesssim 1/(N \epsilon^3)$, $F = O(1/(N \epsilon^3))$, etc. are defined similarly.


\section{Behavior of Trajectories near the Boundary without the Internal Fields} \label{ModelCase}




In this section, we consider a \emph{model} trajectory ODE system, in which we drop the internal fields, and prove that the external magnetic field $B_{ext,N}$ has a "reflective" effect on charged particles when the internal fields are absent. The model trajectory ODE is given as follows: 
(here $\dot{F}$ means $\frac{\partial F}{\partial s}$ for any function $F$)
\begin{equation}
\label{E:N-mock}
\left\{
\begin{aligned}
&\dot X_N = \hat V_{1,N} \\
&\dot {V}_{1,N} =  + \hat {V}_{2,N} B_{ext, N} (X_N)   \\
&\dot {V}_{2,N} =  - \hat {V}_{1,N} B_{ext, N} (X_N)  
\end{aligned}
\right.
\end{equation}


We fix $N$ such that \eqref{cond-N} holds. For $x \in \supp B_{ext, N}$ and any $t$, $v_1$, $v_2$ with $(t, x, v_1, v_2) \in \supp f^N$, consider the trajectory $(X_N (s), V_{1,N} (s), V_{2,N} (s))$ given by \eqref{E:N-mock} and takes the value $(x, v_1, v_2)$ at the time $s=t$. There exists a maximal time interval $I_0 = I_0 (t, x, v_1, v_2)$ that contains $t$, and on which $X_N (s; t, x, v_1, v_2)$ lies in $\supp B_{ext, N}$. We define a \emph{reflection time} $t^*$ for each $(t, x, v_1, v_2)$ (with $x \in \supp B_{ext, N}$, $t \in [0, T]$), which is the time at which $V_1$ changes from $v_1$ to $-v_1$.

\begin{lemma} \label{t1t2mapping-model}
Fix $\epsilon \in (0, 1)$. Let $t$, $x$, $v_1$, $v_2$, $(X_N, V_{1,N}, V_{2,N})$ and $I_0$ be as stated in the last paragraph above, and $x \in (0, \frac{1}{N} - \frac{1}{N} \epsilon] \cup [ 1- (\frac{1}{N} - \frac{1}{N} \epsilon), 1)$. Fix $t$, $x$, $v_1$, $v_2$, there exists a unique $t^*$ in the same interval $I_0$ such that 
\begin{equation} \label{E:t1t2mapping-model-eq0}
(X_N, V_{1,N}, V_{2,N}) (t;t,x,v_1, v_2) = (X_N, -V_{1,N}, V_{2,N}) (t^*;t,x,v_1, v_2) = ( x, v_1, v_2) .
\end{equation}
Moreover, $t^* -t$ only depends on $(x, v_1, v_2)$ and $|t^* -t| \lesssim  \frac{1}{N \epsilon}$. 
For any fixed $(x, v_1, v_2)$, $t \mapsto t^*$ as a function of $t$ is $C^\infty$ and invertible. The Jacobian of the inverse mapping $t^* \mapsto t$ is $|J_N |= | \frac{\partial t}{\partial t^*} | = |J_N (x, v_1, v_2)| =1$.  
\end{lemma}

\begin{flushleft}
\textit{Remark} We call $t^*$ the \emph{reflection time} corresponding to $(t, x, v_1, v_2)$. Notice that Lemma \ref{t1t2mapping-model} only concerns about the behavior of the particle trajectory on $I_0$.  
\end{flushleft}





\begin{proof}




Let us consider the boundary $x=0$. Dropping the $N$ subscript for $(X_N, V_{1,N}, V_{2,N})$ in this lemma and passing to polar coordinates for $V$:
$$V_1 = R \cos \Phi \,, \qquad V_2 = R \sin \Phi$$
Then we check that for a solution to \eqref{E:N-mock}
$$ \frac{d}{ds}(R^2) = 2 V_1 \dot{V}_1 + 2 V_2 \dot{V}_2 = 0 $$
by substituting the equations in \eqref{E:N-mock} for $\dot{V}_1$ and $\dot{V}_2$.  Thus $R$ is constant on $I_0$ and we find that \eqref{E:N-mock} becomes
\begin{equation}
\label{E:mock-2}
\left\{ 
\begin{aligned}
&\dot{X} = \frac{R \cos \Phi}{\sqrt{1+R^2}}\\
&\dot{\Phi} = - \frac{1}{\sqrt{1+R^2}} N  \Psi' (NX)
\end{aligned}
\right.
\end{equation}
Recall that $ \Psi'(Y) \lesssim - \epsilon   < 0$ for $Y = NX \in (0, 1- \epsilon]$. We have $\dot{\Phi}  > 0$.



Since the trajectory is in $C^1$ and $\dot{\Phi}  > 0$ when $s \in I_0$, $\Phi(s)$ evolves in the direction of increasing angle. Let us discuss first the case when $V_1(t)< 0$ (that is, $\Phi (t) \in (\pi/2, 3\pi/2)$). Let $s_{\text{turn}} :=  \min \{ s>t : \Phi (s) = \frac{3\pi}{2}\}$, whose existence is guaranteed by $\dot{\Phi}  > 0$. Since $\Phi$ keeps increasing, $s_{\text{turn}}$ is the unique time in $I_0$ such that $\Phi (s_{\text{turn}}) = 3\pi/2$, $V_1 (s_{\text{turn}}) = 0$, and hence $X$ reaches its minimum at $s = s_{\text{turn}}$. Continuing after $s_{\text{turn}}$, again due to $\dot{\Phi}  > 0$, there exists a unique $t^*$ defined by
$$t^* := \min \{ s>t : \Phi (s) =  3\pi- \Phi (t)\} $$ 
in $I_0$ such that $ \Phi (t^*) = 3\pi-\Phi(t)$. This gives a unique $t^*$ in the interval $I_0$ such that $ (V_{1}, V_{2}) (t;t,x,v_1, v_2) =  (-V_{1}, V_{2}) (t^*;t,x,v_1, v_2) $. Here we used the fact that $\sqrt{V_{1} (s)^2 + V_{2}(s)^2}  =  R(s) \equiv const.$ on $I_0$. After the time $t^*$, $ V_1  (s;t,x,v_1, v_2)  > 0$. Notice that the region $\supp_x B_{ext, N}$ is of size $O(\frac{1}{N})$, which tells us that when $N$ is large enough, the trajectory $(X  , V_1 , V_2  ) (s;t,x,v_1, v_2)$ can exit $\supp_x B_{ext, N}$ within a time period of order $O(\frac{1}{N})$. Therefore, $V_1 (s;t,x,v_1, v_2)$ can only change its sign once in $I_0$.




\begin{center}
\includegraphics[width=1\textwidth]{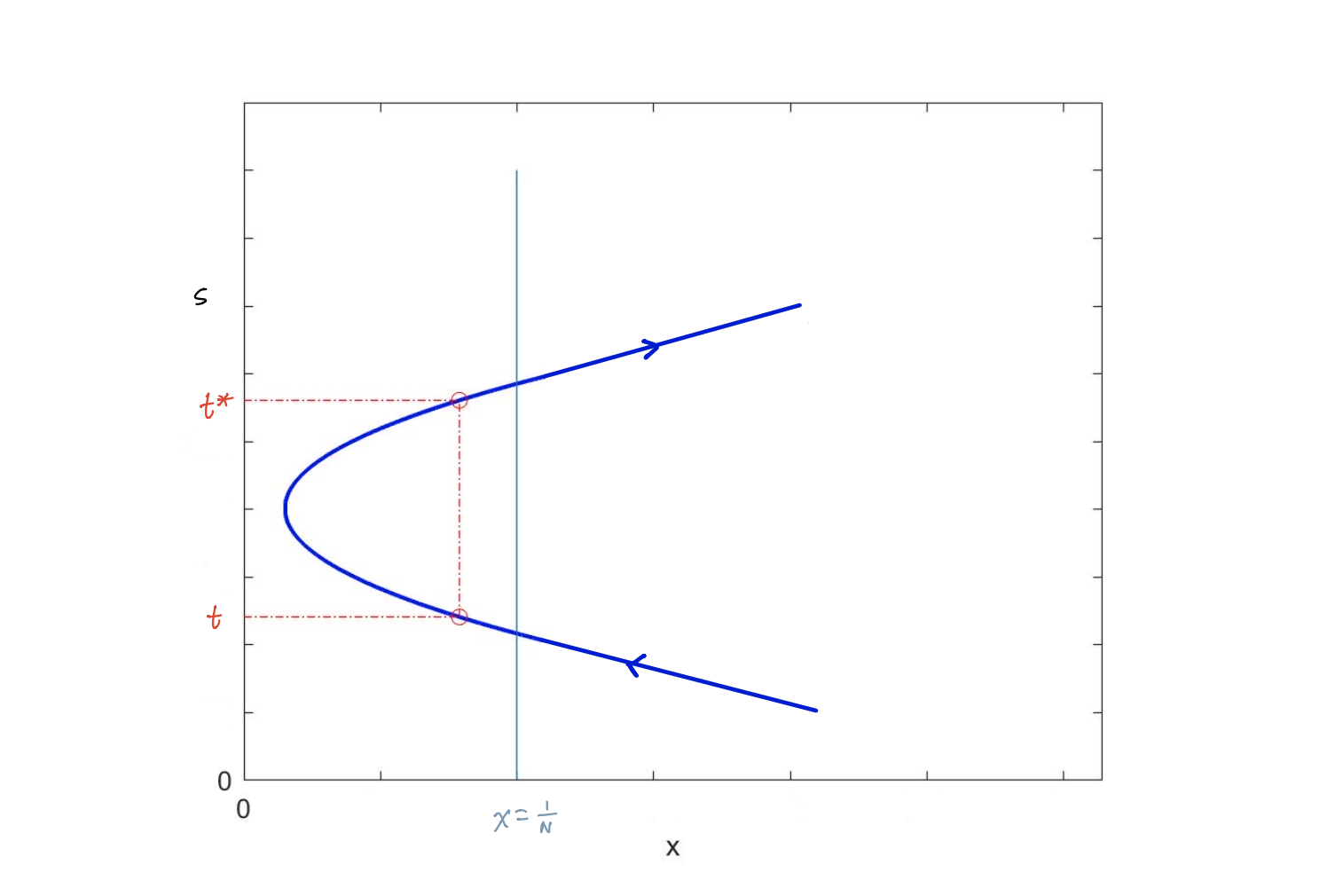}
\end{center}






Cross multiplying the two equations in \eqref{E:mock-2} yields
$$-  N\Psi'(NX) \dot{X} = R(\cos\Phi) \dot{\Phi} $$
and integrating yields
\begin{equation}
\label{E:mock-energy}
  \Psi(NX) - \Psi(NX_1) = R(\sin \Phi_1 - \sin \Phi)
\end{equation}
where $\Phi_1$ and $X_1$ denote the values of $\Phi$ and $X$ at time $t$, respectively. Let $\Phi_2$ and $X_2$ denote the value of $\Phi$ and $X$ at time $t^*$, respectively. It follows that $X_2 = X_1$ since $\Psi(NX)$ is monotone for $X \in (0, 1/N]$, and therefore
$$ (X, V_{1}, V_{2}) (t;t,x,v_1, v_2) = (X, -V_{1}, V_{2}) (t^*;t,x,v_1, v_2) . $$

On the other hand,
\begin{align*}
\dot{\Phi} &= \frac{-N}{\sqrt{1+R^2}}  \Psi'(NX)  \\
\end{align*}
implies
\begin{equation*}
\frac{-N}{\sqrt{1+R^2}} ds = \frac{1}{\Psi'(NX)} d \Phi . \\
\end{equation*}
Integrating yields
\begin{equation} 
\frac{-N}{\sqrt{1+R^2}} (t^* - t) = \int^{3\pi- \Phi_1}_{\Phi_1} \frac{1}{\Psi'(NX)} d \Phi , \\
\end{equation}
which gives
\begin{equation}  \label{E:mock-3}
 t^* - t = \frac{\sqrt{1+R^2}}{-N} \int^{3\pi- \Phi_1}_{\Phi_1} \frac{1}{\Psi'(NX)} d \Phi . \\
\end{equation}
Since $ \Psi'(Y) < 0$, $  |\Psi'(Y) | \gtrsim \epsilon$ for $Y = NX \in (0, 1 -\epsilon]$ and $|V(s)| \leq k_0 + C_1 T$ (Lemma \ref{NS3-lm3.1}), there holds
\begin{equation*}
|t^* - t| \lesssim \frac{1}{N} \sqrt{1+(k_0 + C_1 T)^2} \cdot 2 \pi  \frac{1}{\epsilon} = \frac{2 \pi}{N \epsilon} \sqrt{1+(k_0 + C_1 T)^2} \lesssim \frac{1}{N \epsilon} .
\end{equation*} 
From \eqref{E:mock-3} we learn that for fixed $(x, v_1, v_2)$ with $x \in (0, 1/N]$, $t^* -t$ only depends on $(x, v_1, v_2)$ and is of $O(\frac{1}{N \epsilon})$.
Hence the mapping $t \mapsto t^*$ is invertible. 
Moreover, we have 
$$ \frac{\partial t^*}{\partial t} =1   $$
so the Jacobian of the inverse mapping $ t^* \mapsto t$ is $|J_N | =| \frac{\partial t}{\partial t^*} |= |J_N (x, v_1, v_2)| =1$.


For the case $V_1 (t) > 0$ we define
$$t^* :=  \max \{ s<t : \Phi (s) =  3\pi- \Phi (t)\} $$
Then all the properties in the statement of the lemma hold for this $t^*$:
\begin{equation}  
(X, V_{1}, V_{2}) (t;t,x,v_1, v_2) = (X, -V_{1}, V_{2}) (t^*;t,x,v_1, v_2) = ( x, v_1, v_2) .
\end{equation}
Moreover, $t^* -t$ only depends on $(x, v_1, v_2)$ and $|t^* -t| \lesssim  \frac{1}{N \epsilon}$. 
For any fixed $(x, v_1, v_2)$, $t \mapsto t^*$ as a function of $t$ is $C^\infty$ and invertible. The Jacobian of the inverse mapping $ t^*  \mapsto t$ is $|J_N |= | \frac{\partial t}{\partial t^*} | = |J_N (x, v_1, v_2)| =1$. 

The case $V_1 (t) =0$ is trivial: We simply take $t^* = t$ and the properties in the statement of the lemma hold. 




For the boundary $x =1$ (that is, $x \in [ 1- (\frac{1}{N} - \frac{1}{N} \epsilon), 1)$), the mapping $t \mapsto t^*$ is defined similarly, making use of
$$ \psi_{ext, N} (x) = \Psi (N(1-x))  $$
for $x$ close to $1$. 

To summarize, we define $t^*$ as
$$t^*  := \min \{ s>t : \Phi (s) =  3\pi- \Phi (t)\} $$ 
when $V_{1}(t) < 0$, $x \in (0, \frac{1}{N} \epsilon]$ or $V_{1}(t) > 0$, $x \in [ 1- (\frac{1}{N} - \frac{1}{N} \epsilon), 1)$, and 
$$t^*  := \max \{ s<t : \Phi (s) =  3\pi- \Phi (t)\} $$
when $V_{1}(t) > 0$, $x \in (0, \frac{1}{N} \epsilon]$ or $V_{1}(t) < 0$, $x \in [ 1- (\frac{1}{N} - \frac{1}{N} \epsilon), 1)$. Then all the properties in the statement of the lemma hold. 

\end{proof}

From Lemma \ref{t1t2mapping-model}, we track the trajectory backwards in time and deduce

\begin{corollary} \label{cor-t1t2mapping-model}
Let $t$, $x$, $v_1$, $v_2$ be as in Lemma \ref{t1t2mapping-model}, then
\begin{equation} \label{E:cor-t1t2mapping-model-eq0}
(X_N, V_{1,N}, V_{2,N}) (0;t,x,v_1, v_2) = (X_N, V_{1,N}, V_{2,N}) (0;t^*,x,-v_1, v_2) . 
\end{equation}
\end{corollary}

\begin{proof}
By \eqref{E:t1t2mapping-model-eq0}, we have
$$  (X_N, V_{1,N}, V_{2,N}) (t;t,x,v_1, v_2) = (X_N, V_{1,N}, V_{2,N}) (t;t^*,x,-v_1, v_2)  = ( x, v_1, v_2) . $$
Let $L_1$ denote the trajectory on which $(t,x,v_1, v_2)$ lies and $L_2$ denote the trajectory on which $(t^*,x, -v_1, v_2)$ lies. Since the trajectory taking the value $(x, v_1, v_2)$ at time $t$ is unique, we learn from the equality above that $L_1$ and $L_2$ are identical. Take the value of $L_1$ and $L_2$ at time $0$, we obtain
$$ (X_N, V_{1,N}, V_{2,N}) (0;t,x,v_1, v_2) = (X_N, V_{1,N}, V_{2,N}) (0;t^*,x,-v_1, v_2) . $$ 
\end{proof}



\section{Behavior of Trajectories near the Boundary with the Internal Fields}
\label{GeneralCase}





In this section, we analyze the behavior of the trajectory corresponding to \eqref{VlasovBext} near $\partial \Omega$. To this end, we make use of the results obtained in Section \ref{ModelCase}, which described the reflecting behavior of the \emph{model} trajectory near $\partial \Omega$.  

We fix $N$ such that \eqref{cond-N} holds. For $x \in \supp B_{ext, N}$ and any $t$, $v_1$, $v_2$ with $(t, x, v_1, v_2) \in \supp f^N$, consider the trajectory $(X_N (s), V_{1, N} (s), V_{2, N} (s))$ given by the following ODE system corresponding to \eqref{VlasovBext}: (here $\dot{F}$ means $\frac{\partial F}{\partial s}$)
\begin{equation}
\label{E:N-real}
\left\{
\begin{aligned}
&\dot X_N = \hat V_{1,N} \\
&\dot V_{1,N} =  E^N_1 (s, X_N) + \hat V_{2,N} B^N (s, X_N) + \hat V_{2,N} B_{ext, N} (X_N) \\
&\dot V_{2,N} =  E^N_2 (s, X_N) - \hat V_{1,N} B^N (s, X_N) - \hat V_{1,N} B_{ext, N} (X_N)
\end{aligned}
\right.
\end{equation}
and takes the value $(x, v_1, v_2)$ at the time $s=t$. Let $I_0 = I_0 (t, x, v_1, v_2)$ denote the maximal time interval which contains $t$ and on which the trajectory $X_N (s; t, x, v_1, v_2)$ stays in $\supp B_{ext, N}$.

Now, we turn off the internal electromagnetic field in the trajectory described above for the part when $s \in I_0$ and denote the corresponding trajectory as $(\Xd_N, \Vd_{1, N}, \Vd_{2, N})$. That is, let $(\Xd_N (s), \Vd_{1, N} (s), \Vd_{2, N} (s))$ be a trajectory that also takes the value $(x, v_1, v_2)$ at the time $s=t$, determined by the following ODE system:
\begin{equation}
\label{E:N-mock-1}
\left\{
\begin{aligned}
&\dXd_N = \hVd_{1,N} \\
&\dVd_{1,N} = \big[ E^N_1 (s, \Xd_N) + \hVd_{2,N} B^N (s, \Xd_N) \big] \mathbf{1}_{s \notin I_0} (s)  + \hVd_{2,N} B_{ext, N} (\Xd_N)    \\
&\dVd_{2,N} = \big[ E^N_2 (s, \Xd_N) - \hVd_{1,N} B^N (s, \Xd_N) \big]  \mathbf{1}_{s \notin I_0} (s)  - \hVd_{1,N} B_{ext, N} (\Xd_N)  
\end{aligned}
\right.
\end{equation}
Let $\Id_0 = \Id_0 (t, x, v_1, v_2)$ denote the maximal interval which contains $t$ and on which the trajectory $\Xd_N (s; t, x, v_1, v_2)$ stays in $\supp B_{ext, N}$.

For $(\Xd_N, \Vd_{1, N}, \Vd_{2, N})$, $x \in \supp B_{ext, N}$ and $(t, x, v_1, v_2) \in \supp f^N$, Lemma \ref{t1t2mapping-model} and Corollary \ref{cor-t1t2mapping-model} still apply: 
There exists a unique reflection point $t^* \in \Id_0$, which is defined by 
$$t^*  :=  \min \{ s>t : \Phi (s) =  3\pi- \Phi (t)\} $$ 
when $\Vd_{1,N}(t) < 0$, $x \in (0, 1/N]$ or $\Vd_{1,N}(t) > 0$, $x \in [1-1/N, 1)$, and 
$$t^*  := \max \{ s<t : \Phi (s) =  3\pi- \Phi (t)\} $$
when $\Vd_{1,N}(t) > 0$, $x \in (0, 1/N]$ or $\Vd_{1,N}(t) < 0$, $x \in [1-1/N, 1)$. The reflection point $t^*$ satisfies that $ (\Xd_N, \Vd_{1,N}, \Vd_{2,N}) (t;t,x,v_1, v_2) = (\Xd_N, -\Vd_{1,N}, \Vd_{2,N}) (t^*;t,x,v_1, v_2) = ( x, v_1, v_2) $,  
and we have
\begin{equation} \label{t1t2mapping-real-error-eq0}
(\Xd_N, \Vd_{1,N}, \Vd_{2,N}) (0;t,x,v_1, v_2) = (\Xd_N, \Vd_{1,N}, \Vd_{2,N}) (0;t^*,x,-v_1, v_2) . 
\end{equation}


Recall that (without loss of generality) we assume $N$ is large enough such that 
$$dist (\supp_x f_0 (x, v), \supp_x B_{ext, N}(x) ) >0 . $$
Hence for each $(t, x, v_1, v_2)$ with $x \in \supp_x B_{ext, N}$, $t \geq 0$ and $f^N ( t, x, v_1, v_2) = f_0 ((X_N, V_N) (0; t, x, v_1, v_2)) \neq 0$, there must hold $t > 0$, and moreover $t^* > 0 $ because 
$$ (\Xd_N, \Vd_{1,N}, \Vd_{2,N}) (t;t,x,v_1, v_2) = (\Xd_N, \Vd_{1,N}, \Vd_{2,N}) (t;t^*,x,-v_1, v_2) . $$
Furthermore, we assume $N \geq 8 $ is large enough, so for each $(t, x, v_1, v_2)$ with $x \in \supp_x B_{ext, N}$, $t \in [0, T]$, $f^N ( t, x, v_1, v_2) \neq 0$, there holds $t^* (t, x, v_1, v_2) \in [0, 2T]$. 




\begin{lemma} \label{t1t2mapping-real-error}
Fix $\epsilon \in (0, 1)$. For any $x \in (0, \frac{1}{N} - \frac{1}{N} \epsilon] \cup [ 1- (\frac{1}{N} - \frac{1}{N} \epsilon), 1)$
and any $t$, $v_1$, $v_2$ with $(t, x, v_1, v_2) \in \supp f^N $, denote $\zeta := (x, v_1, v_2)$ and consider the path $(X_N, V_{1,N}, V_{2,N})$ given by \eqref{E:N-real} which takes the value $\zeta = (x, v_1, v_2)$ at time $t$.
Let $t^* = t^* (t, x, v_1, v_2)$ be as defined above. \\
1) There holds 
\begin{equation} \label{t1t2mapping-real-error-eq1}  
|(  -V_{1,N}, V_{2,N}) (t^*;t, \zeta) - (  V_{1,N}, V_{2,N}) (t;t, \zeta) | \lesssim \frac{1}{N \epsilon} , \   |X_N (t^*; t,\zeta) -  X_N  ( t;t, \zeta)  | \lesssim \frac{1}{N^2 \epsilon} . 
\end{equation}
2) Take 
$$\tilde{x}  :=  X_N (t^*;t, \zeta), \ \tilde{v}_1  := -V_{1,N} (t^*; t,\zeta), \ \tilde{v}_2  := V_{2,N} (t^*;t,\zeta) , $$
and denote $\tilde{\zeta}  := (\tilde{x}, -\tilde{v}_1, \tilde{v}_2)$, \eqref{t1t2mapping-real-error-eq1} can be equivalently written as 
$$|\tilde{x} - x| \lesssim \frac{1}{N^2 \epsilon} , \ |\tilde{v}_1 - v_1 | \lesssim \frac{1}{N \epsilon}, \ |\tilde{v}_2 - v_2 | \lesssim \frac{1}{N \epsilon}. $$
Moreover, going backwards in time, we have
\begin{equation}  \label{t1t2mapping-real-error-eq2}
\begin{split}
&  X_N  (0;t,\zeta) = X_N (0;t^*,\tilde{\zeta} ), \\
& V_{1,N} (0;t,\zeta) = V_{1,N} (0;t^*,\tilde{\zeta} )  , \\
& V_{2,N} (0;t,\zeta) = V_{2,N} (0;t^*,\tilde{\zeta} ) . \\
\end{split}
\end{equation}
3) Fix $t \in [0, T]$, the Jacobian of the inverse mapping $ (\tilde{x}, \tilde{v}_1 , \tilde{v}_2) \mapsto (x, v_1, v_2) $ (denoted by $|\mathcal{J}_N| = |\frac{\partial (x, v_1, v_2)}{\partial (\tilde{x}, \tilde{v}_1 , \tilde{v}_2)}| $) satisfies
\begin{equation}
\big||\mathcal{J}_N| -1 \big| \lesssim \frac{1}{N \epsilon^3} . 
\end{equation} 
The constants in the $\lesssim$'s in this lemma only depends on $\Psi$, $f_0$, $\|(E^N, B^N)\|_{C^1_{t, x} ([0, 2T] \times \Omega)}$, $C_v$ and $T$, and therefore only depends on the initial-boundary data, $\Psi$ and $T$ (see Section \ref{SectionWL}).
\end{lemma}


\begin{center}
\includegraphics[width=1\textwidth]{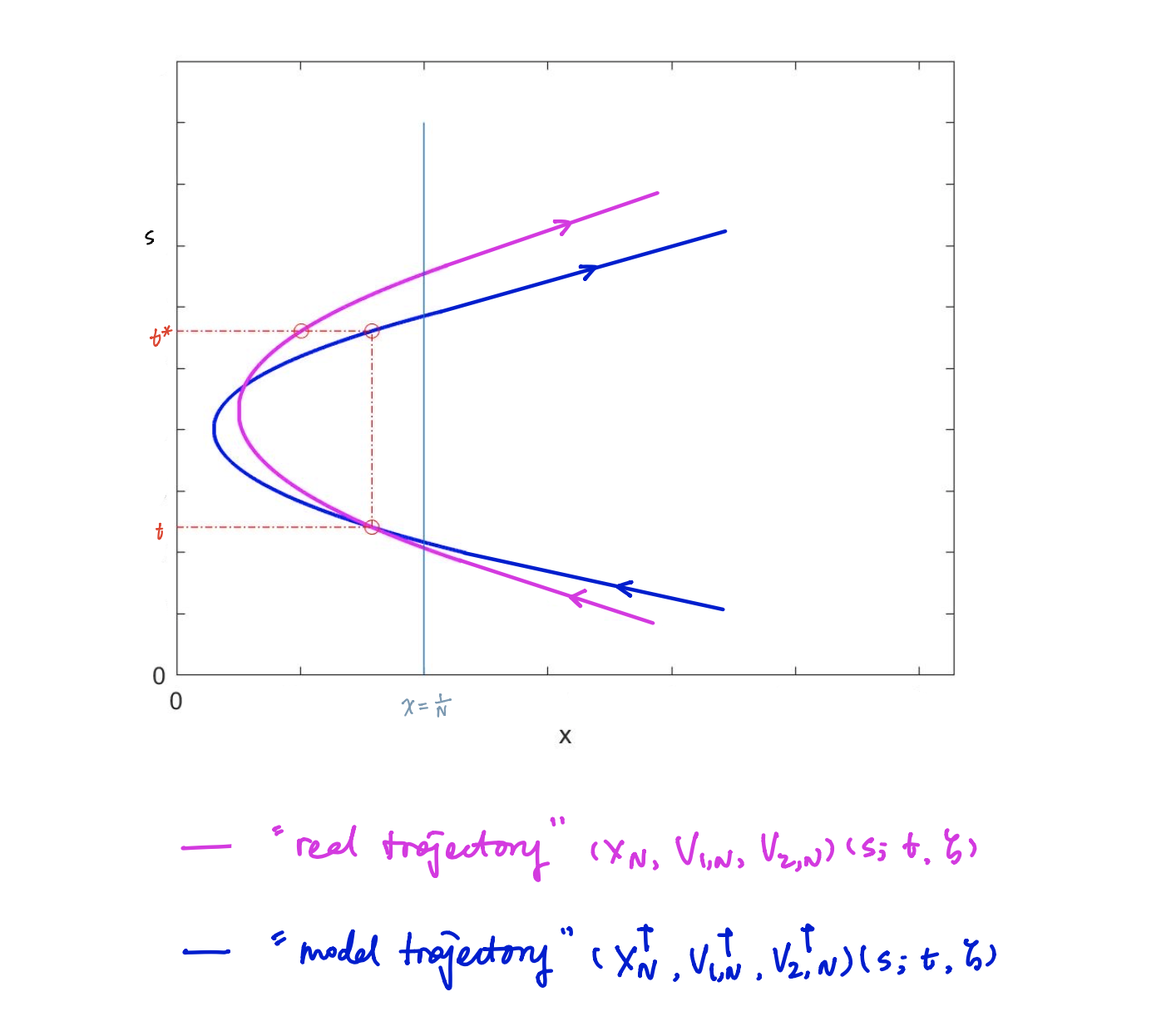}
\end{center}




\begin{proof}
It suffices to consider the boundary $x=0$ and the corresponding region $\{ x: dist(x, \partial \Omega) \leq \frac{1}{N}\}$ since the boundary $x=1$ is similar. Notice that $|t^* - t | \lesssim \frac{1}{N \epsilon}$ for all $x \in (0, \frac{1}{N} - \frac{1}{N} \epsilon] \cup [ 1- (\frac{1}{N} - \frac{1}{N} \epsilon), 1)$ and any $t$, $v_1$, $v_2$ with $(t, x, v_1, v_2) \in \supp f^N $ (Lemma \ref{t1t2mapping-model}).

We first prove 1). 
Our strategy is to compare the values of the two trajectories $(X_N, V_{1,N}, V_{2,N})(s; t, \zeta)$ and $(\Xd_N, \Vd_{1,N}, \Vd_{2,N})(s; t, \zeta)$ at the time $s=t^*$ (here $(\Xd_N, \Vd_{1,N}, \Vd_{2,N})$ is defined as stated in the beginning of the section). We already know from the assumptions and Lemma \ref{t1t2mapping-model} that    
\begin{equation}  \label{t1t2mapping-real-error-eq5}  
 (X_N, V_{1,N}, V_{2,N}) (t; t, \zeta) = (\Xd_N, \Vd_{1,N}, \Vd_{2,N}) (t; t, \zeta) = (\Xd_N, -\Vd_{1,N}, \Vd_{2,N}) ( t^*; t, \zeta ) . 
\end{equation}


We introduce the rescaling
\begin{equation}
\begin{split}
& X_N(s) = N^{-1}Y_N(N(s-t)), \ \Xd_N(s) = N^{-1}\Yd_N(N(s-t)) , \\
& V_{1,N}(s) = W_{1,N}(N(s-t)), \ \Vd_{1,N}(s) = \Wd_{1,N}(N(s-t)) , \\
& V_{2,N}(s) = W_{2,N}(N(s-t)) , \ \Vd_{2,N}(s) = \Wd_{2,N}(N(s-t)) , \\
\end{split}
\end{equation}
and let $\sigma=N(s-t)$.
\eqref{E:N-real} becomes 
\begin{equation}
\label{E:N-real-rescaled}
\left\{
\begin{aligned}
& \frac{dY_N}{d\sigma}   = \hat W_{1,N} \\
& \frac{dW_{1,N}}{d \sigma} = \frac{1}{N} \big[ E^N_1 (N^{-1} \sigma +t, N^{-1} Y_N) + \hat W_{2,N} B^N (N^{-1} \sigma +t, N^{-1} Y_N) \big]  + \hat W_{2,N} \partial_y \Psi(Y_N) \\
& \frac{dW_{2,N}}{d \sigma} = \frac{1}{N} \big[ E^N_2 (N^{-1} \sigma +t, N^{-1} Y_N) - \hat W_{1,N} B^N (N^{-1} \sigma +t, N^{-1} Y_N) \big]   - \hat W_{1,N} \partial_y \Psi(Y_N)
\end{aligned}
\right.
\end{equation}
and \eqref{E:N-mock-1} becomes
\begin{equation}
\label{E:N-mock-rescaled-origin}
\left\{
\begin{aligned}
& \frac{d\Yd_N}{d\sigma} = \hWd_{1,N} \\
& \frac{d\Wd_{1,N}}{d\sigma} = \frac{1}{N} \big[ E^N_1 (N^{-1} \sigma +t, N^{-1} \Yd_N) + \hWd_{2,N} B^N (N^{-1} \sigma +t, N^{-1} \Yd_N) \big] \mathbf{1}_{\{\sigma: s = N^{-1} \sigma +t \notin I_0\}}  \\
& \quad \quad \qquad + \hWd_{2,N} \partial_y \Psi(\Yd_N) \\
& \frac{d\Wd_{2,N}}{d\sigma} = \frac{1}{N} \big[ E^N_2 (N^{-1} \sigma +t, N^{-1} \Yd_N) - \hWd_{1,N} B^N (N^{-1} \sigma +t, N^{-1} \Yd_N) \big] \mathbf{1}_{\{\sigma: s = N^{-1} \sigma +t \notin I_0\}}  \\
& \quad \quad \qquad - \hWd_{1,N} \partial_y \Psi(\Yd_N)
\end{aligned}
\right.
\end{equation}
In the time interval $\{ \sigma: s = N^{-1} \sigma +t \in I_0 \}$, the indicator function in \eqref{E:N-mock-rescaled-origin} is $0$. Hence in this interval \eqref{E:N-mock-rescaled-origin} becomes
\begin{equation}
\label{E:N-mock-rescaled}
\left\{
\begin{aligned}
& \frac{d\Yd_N}{d\sigma} = \hWd_{1,N} \\
& \frac{d\Wd_{1,N}}{d\sigma}  =  + \hWd_{2,N} \partial_y \Psi(\Yd_N) \\
& \frac{d\Wd_{2,N}}{d\sigma}  =  - \hWd_{1,N} \partial_y \Psi(\Yd_N)
\end{aligned}
\right.
\end{equation}
Now fix $t$ and $\zeta= (x, v_1, v_2)$. For any $s \in \mathbb{R}$ and $\sigma=N(s-t)$, let
\begin{equation}
\begin{split}
& Y_N (\sigma) = NX_N (s; t, \zeta), \ W_{1,N} (\sigma) = V_{1,N} (s; t, \zeta), \ W_{2,N} (\sigma) = V_{2,N} (s; t, \zeta), \\
& \Yd_N (\sigma) = N\Xd_N (s; t, \zeta), \ \Wd_{1,N} (\sigma) = \Vd_{1,N} (s; t, \zeta), \ \Wd_{2,N} (\sigma) = \Vd_{2,N} (s; t, \zeta). \\
\end{split}
\end{equation} 
We denote
$$y = N x , \ w_1 = v_1, \ w_2 = v_2, \ \tilde{y} = N \tilde{x}, \ \tilde{w}_1 = \tilde{v}_1, \ \tilde{w}_2 = \tilde{v}_2. $$
(see the picture below)
\begin{center}
\includegraphics[width=1\textwidth]{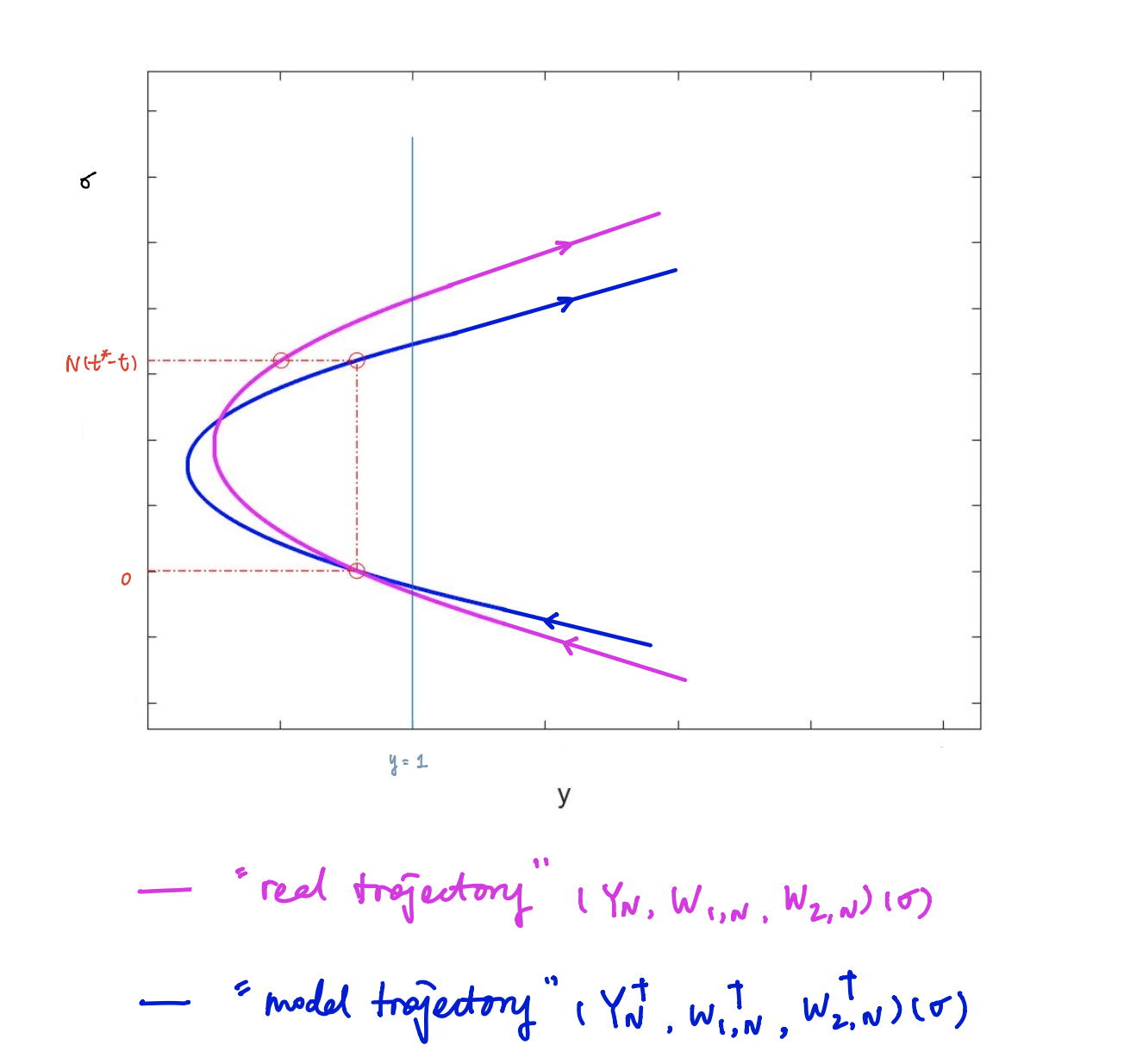}
\end{center}
By the definitions of $(Y_N, W_{1,N}, W_{2,N})$, $(\Yd_N, \Wd_{1,N}, \Wd_{2,N})$, $\sigma$ and $\tau^*$, together with the values of $(X_N, V_{2,N}, V_{2,N})$ and $(\Xd_N, \Vd_{2,N}, \Vd_{2,N})$ at time $t$ and $t^*$, we have
\begin{equation}  \label{t1t2mapping-real-error-eq13} 
\begin{split}
& Y_N (0) = \Yd_N (0)= \Yd_N (N(t^*-t)) = y , \ W_{1,N}(0) = \Wd_{1,N}(0) = -\Wd_{1,N}(N(t^*-t)) = w_1, \\
& W_{2,N}(0) = \Wd_{2,N}(0) =\Wd_{2,N}(N(t^*-t))= w_2 , \\
& Y_N (N(t^*-t)) =  \tilde{y},  \ W_{1,N}(N(t^*-t)) = - \tilde{w}_1 ,  \ W_{2,N}(N(t^*-t)) = \tilde{w}_2 . \\
\end{split}  
\end{equation}

From now on in this section we drop the subscript $N$ in $(Y_N, W_{1,N}, W_{2,N})$ and $(\Yd_N, \Wd_{1,N}, \Wd_{2,N})$ for simplicity. 

Notice that \eqref{t1t2mapping-real-error-eq1} is equivalent to 
\begin{equation} \label{t1t2mapping-real-error-eq15}  
|Y (N(t^*-t)) - Y(0)| \lesssim \frac{1}{N \epsilon} , \ |W_1(N(t^*-t)) + W_1 (0) | \lesssim \frac{1}{N \epsilon} , \ |W_2(N(t^*-t)) -W_2 (0) | \lesssim \frac{1}{N \epsilon} .
\end{equation} 
By \eqref{t1t2mapping-real-error-eq13}, \eqref{t1t2mapping-real-error-eq15} is equivalent to
\begin{equation}  \label{t1t2mapping-real-error-eq14}
\begin{split}
& |Y (N(t^*-t)) - \Yd (N(t^*-t))| \lesssim \frac{1}{N \epsilon} , \\
& |W_1(N(t^*-t)) - \Wd_1 (N(t^*-t))| \lesssim \frac{1}{N \epsilon} , \\
& |W_2(N(t^*-t)) - \Wd_2 (N(t^*-t))| \lesssim \frac{1}{N \epsilon} , \\
\end{split}
\end{equation}
which is what we shall prove now.

We use $F'$ to denote $\frac{\partial F}{\partial \sigma}$ for any function $F$. We take the difference of \eqref{E:N-real-rescaled} and \eqref{E:N-mock-rescaled}, and estimate the $L^\infty$ norm of the right hand side of the resulting ODEs by using the uniform boundedness (in $N$) of $(E^N, B^N)$ as well as the fact that for $i= 1,2$, 
\begin{equation} \label{E:N-difference-rescaled-est2}
\begin{split}
& |\hat W_i \big[ \partial_y \Psi(Y) - \partial_y \Psi(\Yd) \big]| \leq |\partial^2_y \Psi (y_0) | |Y- \Yd| , \\
& | \big[ \hat W_i - \hWd_i \big] \partial_y \Psi(\Yd)| \lesssim \big[ |W_1 - \Wd_1| + |W_2 - \Wd_2| \big] |\partial_y \Psi (y_0) | ,  \\
\end{split}
\end{equation}
which follows by the Mean-Value Theorem and Corollary \ref{NS3-lm3.1-cor}. Applying Gronwall's inequality on the difference ODEs of \eqref{E:N-real-rescaled} and \eqref{E:N-mock-rescaled} on the time interval $[0, N(t^*-t)]$ (notice that $N(t^*-t) \lesssim  \epsilon^{-1}$) yields \eqref{t1t2mapping-real-error-eq14}, which gives \eqref{t1t2mapping-real-error-eq1}. 1) is proved. Notice that by a similar process as above we also obtain
\begin{equation} \label{t1t2mapping-real-error-1-cor}
|Y- \Yd|\lesssim \frac{1}{N \epsilon} \text{ for all } \sigma \in [0, N(t^*-t)] .
\end{equation}

Next we prove 2). Take 
$$\tilde{x} := X_N (t^*; t, \zeta), \ \tilde{v}_1  := -V_{1,N} (t^*; t, \zeta), \ \tilde{v}_2  := V_{2,N} (t^*; t, \zeta) . $$
Then \eqref{t1t2mapping-real-error-eq1} can be written as 
$$|\tilde{x} - x| \lesssim \frac{1}{N^2 \epsilon} , \ |\tilde{v}_1 - v_1 | \lesssim \frac{1}{N \epsilon} , \ |\tilde{v}_2 - v_2 | \lesssim \frac{1}{N \epsilon} . $$ 
We obtain \eqref{t1t2mapping-real-error-eq2} immediately by observing that the trajectories $(X_N, V_{1,N}, V_{2,N})(t, \zeta)$ and $(X_N, V_{1,N}, V_{2,N})(t^*, \tilde{\zeta})$ are identical to each other. 2) is verified. 



 


Lastly we prove 3). We want to show that the Jacobian $|\mathcal{J}_N| = |\frac{\partial (x, v_1, v_2)}{\partial (\tilde{x}, \tilde{v}_1 , \tilde{v}_2)}| $ satisfies $\big| |\mathcal{J}_N | - 1 \big| \lesssim \frac{1}{N \epsilon^3}$. It is equivalent to prove that
\begin{equation} \label{t1t2mapping-real-error-JN}
\big| | \frac{\partial (\tilde{x}, \tilde{v}_1 , \tilde{v}_2)}{\partial (x, v_1, v_2)} | - 1 \big| \lesssim \frac{1}{N \epsilon^3} .
\end{equation}
Therefore we just need to derive
\begin{equation} \label{t1t2mapping-real-error-eq8}  
\big| \frac{\partial \tilde{x}}{\partial x}  -1 \big| \lesssim \frac{1}{N \epsilon^3} , \ \big| \frac{\partial \tilde{v}_1}{\partial v_1}  -1 \big| \lesssim \frac{1}{N \epsilon^3} , \ \big| \frac{\partial \tilde{v}_2}{\partial v_2}  -1 \big| \lesssim \frac{1}{N \epsilon^3} ,  
\end{equation} 
and that all the non-diagonal entries in the matrix $\frac{\partial (\tilde{x}, \tilde{v}_1 , \tilde{v}_2)}{\partial (x, v_1, v_2)}$ are of size $O(1/N^3)$ in $L^\infty$. Then in view of the rule of Sarrus, \eqref{t1t2mapping-real-error-JN} holds. 

By rescaling, to prove \eqref{t1t2mapping-real-error-eq8}, it suffices to show
\begin{equation}  \label{t1t2mapping-real-error-eq7}  
\big| \frac{\partial \tilde{y}}{\partial y} -1 \big| \lesssim \frac{1}{N \epsilon^3} , \ \big| \frac{\partial \tilde{w}_1}{\partial w_1}  -1 \big| \lesssim \frac{1}{N \epsilon^3} , \ \big| \frac{\partial \tilde{w}_2}{\partial w_2}  -1 \big| \lesssim \frac{1}{N \epsilon^3},
\end{equation}
and that all the non-diagonal entries in the matrix $\frac{\partial (\tilde{y}, \tilde{w}_1 , \tilde{w}_2)}{\partial (y, w_1, w_2)}$ are of size $O(\frac{1}{N \epsilon^3})$ in $L^\infty$. We first prove $\big| \frac{\partial \tilde{y}}{\partial y}  -1 \big| \lesssim \frac{1}{N \epsilon^3}$.   

We take the derivative of \eqref{E:N-real-rescaled} with respect to the initial condition $y$ to obtain
\begin{equation} \label{t1t2mapping-real-error-eq9}  
\begin{split}
& (\frac{\partial Y}{\partial y}) ' =  \frac{\partial \hat W_1}{\partial y}  ,   \\
& (\frac{\partial W_1}{\partial y})  ' = \frac{1}{N} \big[ N^{-1} \partial_x E^N_1 (N^{-1} \sigma +t, N^{-1} Y ) \frac{\partial Y}{\partial y} + N^{-1} \hat{W}_2 \partial_x B^N (N^{-1} \sigma +t, N^{-1} Y) \frac{\partial Y}{\partial y} \\  
& \quad \quad \qquad + \frac{\partial \hat{W}_2 }{\partial y}   B^N (N^{-1} \sigma +t, N^{-1} Y)  \big]  + \frac{\partial \hat{W}_2}{\partial y} \partial_y \Psi(Y) + \hat{W}_2 \partial_y^2 \Psi(Y) \frac{\partial Y}{\partial y} , \\
& (\frac{\partial W_2}{\partial y})  ' = \frac{1}{N} \big[ N^{-1} \partial_x E^N_2 (N^{-1} \sigma +t, N^{-1} Y ) \frac{\partial Y}{\partial y} - N^{-1} \hat{W}_1 \partial_x B^N (N^{-1} \sigma +t, N^{-1} Y) \frac{\partial Y}{\partial y} \\  
& \quad \quad \qquad - \frac{\partial \hat{W}_1}{\partial y}   B^N (N^{-1} \sigma +t, N^{-1} Y)  \big]  - \frac{\partial \hat{W}_1}{\partial y} \partial_y \Psi(Y) - \hat{W}_1 \partial_y^2 \Psi(Y) \frac{\partial Y}{\partial y}  . \\
\end{split}
\end{equation} 
Similarly, for $\sigma \in \{ \sigma : N^{-1} \sigma +t \in I_0 \}$, we take the derivative of \eqref{E:N-mock-rescaled} with respect to the initial condition $y$ to obtain
\begin{equation}  \label{t1t2mapping-real-error-eq10}  
\begin{split}
& (\frac{\partial \Yd}{\partial y}) ' =  \frac{\partial \hWd_1}{\partial y}  ,   \\
& (\frac{\partial \Wd_1}{\partial y})  ' =   \frac{\partial \hWd_2}{\partial y} \partial_y \Psi(\Yd) + \hWd_2 \partial_y^2 \Psi(\Yd) \frac{\partial \tilde{Y}}{\partial y} , \\
& (\frac{\partial \Wd_2}{\partial y})  ' =   - \frac{\partial \hWd_1}{\partial y} \partial_y \Psi(\Yd) - \hWd_1 \partial_y^2 \Psi(\Yd) \frac{\partial \Yd}{\partial y}  . \\
\end{split}
\end{equation} 
Since $|\partial_y \Psi (Y)| \leq |\partial_y \Psi (y_0)|$, $|\partial_y^2 \Psi (Y)| \leq |\partial_y^2 \Psi (y_0)|$ (see Corollary \ref{NS3-lm3.1-cor}) and same thing holds for $\Yd$, we deduce
\begin{equation} \label{t1t2mapping-real-error-eq12} 
|\frac{\partial Y}{\partial y}| + |\frac{\partial W_1}{\partial y}| +  |\frac{\partial W_2}{\partial y}|+ |\frac{\partial \Yd}{\partial y}| + |\frac{\partial \Wd_1}{\partial y}| + |\frac{\partial \Wd_2}{\partial y}| \lesssim \epsilon^{-1}  \text{ for all } \sigma \in [0, N(t^*-t)] 
\end{equation}
by applying Gronwall's inquality on the time interval $[0, N(t^*-t)]$. (Notice that $\frac{\partial Y}{\partial y}|_{\sigma=0} =\frac{\partial \Yd}{\partial y}|_{\sigma=0} =1$, $\frac{\partial W_1}{\partial y}|_{\sigma=0}=\frac{\partial W_2}{\partial y}|_{\sigma=0} =\frac{\partial \Wd_1}{\partial y}|_{\sigma=0} =\frac{\partial \Wd_2}{\partial y}|_{\sigma=0}  =0$.)

Since $\frac{\partial Y}{\partial y} |_{\sigma =0} = \frac{\partial \Yd}{\partial y} |_{\sigma =0} = \frac{\partial \Yd}{\partial y} |_{\sigma =N(t^*-t)} = 1 $, we have $\big| \frac{\partial \tilde{y}}{\partial y}   -1 \big|  = \big|  \frac{\partial Y}{\partial y} |_{\sigma =N(t^*-t)}  - \frac{\partial \Yd}{\partial y} |_{\sigma =N(t^*-t)} \big| $. It then suffices for us to prove
\begin{equation}
| \frac{\partial (Y- \Yd)}{\partial y} |  \lesssim \frac{1}{N \epsilon^3}
\end{equation}
holds for all $ \sigma \in [0, N(t^*-t)]$. We take the difference of \eqref{t1t2mapping-real-error-eq9} and \eqref{t1t2mapping-real-error-eq10}, and estimate the terms involving the internal fields on the right hand side of the resulting ODEs by using the uniform boundedness (in $N$) of $(E^N, B^N)$. Moreover, we estimate the terms involving $\Psi$ on the right hand side of the resulting ODEs using the triangular inequality, \eqref{t1t2mapping-real-error-eq12}, \eqref{t1t2mapping-real-error-1-cor} and Corollary \ref{NS3-lm3.1-cor} as follows: For $i=1,2$ and $\sigma \in [0, N(t^*-t)]$, 
\begin{equation*}
\begin{split}
| \frac{\partial \hat{W}_i}{\partial y} \partial_y \Psi(Y) -  \frac{\partial \hWd_i}{\partial y} \partial_y \Psi(\Yd) |  
& \leq | \frac{\partial \hat{W}_i}{\partial y} \big[\partial_y \Psi(Y) -   \partial_y \Psi(\Yd) \big] | + | \big[ \frac{\partial \hat{W}_i}{\partial y}  - \frac{\partial \hWd_i}{\partial y} \big] \partial_y \Psi(\Yd) | \\
& \leq | \frac{\partial \hat{W}_i}{\partial y} ||\partial^2_y \Psi(y_0)||Y-\Yd  | + |  \frac{\partial \hat{W}_i}{\partial y}  - \frac{\partial \hWd_i}{\partial y} | |\partial_y \Psi(y_0) | \\
& \lesssim \frac{1}{N \epsilon^2} + | \frac{\partial (W_1 - \Wd_1)}{\partial y} | +| \frac{\partial (W_2 - \Wd_2)}{\partial y} | .  \\ 
\end{split}
\end{equation*} 
Similarly we have, for $i=1,2$ and $\sigma \in [0, N(t^*-t)] $, 
\begin{equation*}
\begin{split}
| \hat{W}_i \partial_y^2 \Psi(Y) \frac{\partial Y}{\partial y}   - \hWd_i \partial_y^2 \Psi(\Yd) \frac{\partial \Yd}{\partial y} |  
& \lesssim \frac{1}{N \epsilon^2} + | \frac{\partial (Y -  \Yd)}{\partial y} | .  \\
\end{split}
\end{equation*}
Applying Gronwall's inequality on the difference ODEs of \eqref{t1t2mapping-real-error-eq9} and \eqref{t1t2mapping-real-error-eq10} on the time interval $[0, N(t^*-t)]$ (notice that $N(t^*-t) \lesssim \epsilon^{-1}$) gives: For all $ \sigma \in [0, N(t^*-t)]$, 
\begin{equation*}
| (\frac{\partial (Y - \Yd )}{\partial y}) | \lesssim \frac{1}{N \epsilon^3} , \ | (\frac{\partial (W_1 - \Wd_1 )}{\partial y})  | \lesssim \frac{1}{N \epsilon^3} , \ | (\frac{\partial (W_2 - \Wd_2 )}{\partial y})  | \lesssim \frac{1}{N \epsilon^3} ,
\end{equation*}
which has what we want. The proof for $\big| \frac{\partial \tilde{y}}{\partial y}  -1 \big| \lesssim \frac{1}{N \epsilon^3} $ is complete.

The proofs for $\big|  \frac{\partial \tilde{w}_1}{\partial w_1}  -1 \big| \lesssim \frac{1}{N \epsilon^3}$ and $ \big| \frac{\partial \tilde{w}_2}{\partial w_2} -1 \big| \lesssim \frac{1}{N \epsilon^3}$ can be carried out similarly by taking derivatives for \eqref{E:N-real-rescaled} and \eqref{E:N-mock-rescaled} with respect to $w_1$ and $w_2$, taking difference, and applying Gronwall's inequality (notice that $\frac{\partial W_1}{\partial w_1} |_{\sigma =0} = \frac{\partial \Wd_1}{\partial w_1} |_{\sigma =0} = \frac{\partial \Wd_1}{\partial w_1} |_{\sigma =N(t^*-t)} = 1 $, and similar things hold for $W_2$, $\Wd_2$ and $w_2$).

The proofs for that the non-diagonal entries in the matrix $\frac{\partial (y, w_1, w_2)}{\partial (\tilde{y}, \tilde{w}_1 , \tilde{w}_2)}$ are of size $O(\frac{1}{N \epsilon^3})$ in $L^\infty$ can be carried out similarly as above, with the observation that $\frac{\partial Y}{\partial w_1} |_{\sigma =0} = \frac{\partial \Yd}{\partial w_1} |_{\sigma =0} = \frac{\partial \Yd}{\partial w_1} |_{\sigma =N(t^*-t)} = 0 $, etc.. 

The proof of the lemma is now complete.


\end{proof}

\section{Proof of the Main Theorem} \label{ProofofMainTheorem}

In this section, we consider the subsequence $(f^N, E^N_1, E^N_2, B^N)$ obtained in Lemma \ref{lm4} and prove Theorem \ref{mtheorem}.  

In particular, $(f^N, E^N_1, E^N_2, B^N)$ is a weak solution of the problem \eqref{VlasovBext} -- \eqref{MaxwellBext} in the sense of Definition \ref{D:weak-Bext}. Hence $f^N$ satisfies
\begin{equation} \label{ProofofMainTheorem-eq1}
\begin{split}
& \int_0^T  \int_{\Omega} \int_{\mathbb{R}^2} f^N \cdot \Big\{ \partial_t \alpha + \hat{v}_1 \partial_x \alpha + ( E^N_1 + \hat{v}_2 B^N+\hat{v}_2 B_{ext, N} ) \partial_{v_1} \alpha  \\
& + ( E^N_2 - \hat{v}_1 B^N- \hat{v}_1 B_{ext, N} ) \partial_{v_2} \alpha  \Big\} \, dv_1 \, dv_2 \,dx \, dt  \\
&  + \int_{\Omega} \int_{\mathbb{R}^2} f^N (0, x, v) \alpha (0, x, v) \, dv_1 \, dv_2 \,dx   = 0   \\
\end{split}
\end{equation}
for any $\alpha (t, x, v) \in C^\infty_c ([0, T) \times \overline{\Omega} \times \mathbb{R}^2)$.

We first prove that the limit $f$ in Lemma \ref{lm4} is a weak solution of the Vlasov equation together with the specular boundary condition \eqref{specularBC} (see Definition \ref{D:weak}). We take the limit $N \to +\infty$ of \eqref{ProofofMainTheorem-eq1}, and notice that
\begin{equation*}
f^N |_{t=0}  = f_0  .
\end{equation*}
Since $f^N  \rightharpoonup f$ in $weak^*-L^\infty (\mathbb{R}_+ \times \Omega \times \mathbb{R}^2)$ and $(E^N, B^N) \rightarrow (E, B)$ strongly in $C^0 ([0, T] \times \Omega)$, we have, for any $\alpha (t, x, v) \in C^\infty_c ([0, T) \times \overline{\Omega} \times \mathbb{R}^2)$,
\begin{equation}
\begin{split}
& \int_0^T  \int_{\Omega} \int_{\mathbb{R}^2} f^N \cdot \Big\{ \partial_t \alpha + \hat{v}_1 \partial_x \alpha + ( E^N_1 + \hat{v}_2 B^N  ) \partial_{v_1} \alpha  + ( E^N_2 - \hat{v}_1 B^N  ) \partial_{v_2} \alpha  \Big\} \, dv_1 \, dv_2 \,dx \, dt  \\
& \rightarrow  \int_0^T  \int_{\Omega} \int_{\mathbb{R}^2} f \cdot \Big\{ \partial_t \alpha + \hat{v}_1 \partial_x \alpha + ( E_1 + \hat{v}_2 B  ) \partial_{v_1} \alpha   + ( E_2 - \hat{v}_1 B  ) \partial_{v_2} \alpha \Big\} \, dv_1 \, dv_2 \,dx \, dt . \\
\end{split}
\end{equation}
The presence of the $B_{ext,N}$ terms yields the extra term
\begin{equation}
\label{E:extra2}
\int_0^T \int_{\Omega} \int_{\mathbb{R}^2} f^N B_{ext, N} (\hat v_2  \partial_{v_1} \alpha - \hat v_1 \partial_{v_2}\alpha) \, dv_1 \, dv_2 \,dx \, dt .
\end{equation}
It suffices to prove that this extra term goes to zero as $N\to \infty$ in order to recover the statement that $f$ satisfies the weak form of the Vlasov equation \eqref{E:weakformVlasov} in Definition \ref{D:weak}.


Notice that $\alpha$ satisfies the properties stated in Definition \ref{D:weak}, and therefore $\hat v_2  \partial_{v_1} \alpha - \hat v_1 \partial_{v_2}\alpha$ is a function in $C^\infty_c ([0, T) \times \overline{\Omega} \times \mathbb{R}^2)$ and it is odd in $v_1$ when $x =0$ or $1$. It suffices to apply the following lemma to our setting with $\tilde{\alpha} = \hat v_2  \partial_{v_1} \alpha - \hat v_1 \partial_{v_2}\alpha $:


\begin{lemma} \label{(extra)-converge}
Let $\tilde{\alpha}$ be an arbitrary function in $C^\infty_c ({\mathbb{R}}  \times \overline{\Omega} \times \mathbb{R}^2)$ that satisfies the following symmetry conditions at the boundary: 
$$\tilde{\alpha} (t, 0, v_1, v_2) = -\tilde{\alpha} (t, 0, -v_1, v_2), \ \tilde{\alpha} (t, 1, v_1, v_2) = -\tilde{\alpha} (t, 1, -v_1, v_2). $$ 
Then
\begin{equation}
\label{E:extra20}
\int_0^T  \int_{\Omega} \int_{\mathbb{R}^2} B_{ext, N} (x) f^N (t, x, v_1, v_2) \tilde{\alpha} (t, x, v_1, v_2) dv_2 dv_1 dx dt \rightarrow 0
\end{equation}
as $N \rightarrow + \infty$.
\end{lemma}


Once Lemma \ref{(extra)-converge} is proved, we can verify that the limit $(f, E_1, E_2, B)$ is a weak solution for the Vlasov equation with the specular boundary condition \eqref{specularBC} on $f$, and thus complete the proof of Theorem \ref{mtheorem}:

\begin{proof} 
{\it{(of Theorem \ref{mtheorem})}}

By the discussions above and Lemma \ref{(extra)-converge}, we deduce that the limit $(f, E_1, E_2, B)$ solves the Vlasov equation in the sense of \eqref{E:weakformVlasov} in Definition \ref{D:weak}. It suffices to carry out the limit process for the Maxwell equations and verify \eqref{E:weakformMaxwell1} -- \eqref{E:weakformMaxwell4} for $(f, E_1, E_2, B)$. Notice that by \eqref{weakformMaxwellN1}, $(f^N, E_1^N, E_2^N, B^N)$ satisfies
\begin{equation} 
\begin{split}
& -\int^1_0  \int^T_0 E^N_1 \partial_t \varphi_1 dt dx  
 - \int^1_0   E_{1, 0} (x) \varphi_1 (0, x) dx  
+  \int^1_0 \int^T_0 \int_{\mathbb{R}^2} \hat{v}_1 f^N  \varphi_1 dv dt dx  =0 .  \\
\end{split}
\end{equation}
Now since $(E_1^N, E_2^N, B^N) \rightarrow (E_1, E_2, B)$ strongly in $C^0_{t, x}$, $f^N \rightharpoonup f$ weakly-* in $L^\infty_{t, x, v}$, we recover \eqref{E:weakformMaxwell1} by taking $N \rightarrow +\infty$. (Here notice that since $\supp_v f^N (t) \subset D_{C_v}$ for all $t \in [0, T]$ and all $N$, the inner integral $\int_{\mathbb{R}^2} \hat{v}_1 f^N  \varphi_1 dv$ in the last term above can be replaced by $\int_{D_{C_v}} \hat{v}_1 f^N  \varphi_1 dv$.) Similarly we obtain \eqref{E:weakformMaxwell2} from \eqref{weakformMaxwellN2} by taking $N \rightarrow +\infty$. Also, by \eqref{weakformMaxwellN3}, $(f^N, E_1^N, E_2^N, B^N)$ satisfies
\begin{equation}  
\begin{split}
& -\int^1_0 \int^T_0  E^N_2 \partial_t \varphi_3 dt dx - \int^1_0 \int_{\mathbb{R}^2} E_{2, 0} (x) \varphi_3 (0, x) dx dv   \\
& - \int^1_0 \int^T_0 B^N \partial_x \varphi_3 dt dx  + \int^T_0   B_b (t) \varphi_3 (t, 1) dt
  + \int^1_0 \int^T_0 \int_{\mathbb{R}^2} \hat{v}_2 f^N \varphi_3 dv  dt dx  =0 . \\
\end{split}
\end{equation}
Again since $(E_1^N, E_2^N, B^N) \rightarrow (E_1, E_2, B)$ strongly in $C^0_{t, x}$, $f^N \rightharpoonup f$ weakly-* in $L^\infty_{t, x, v}$, we recover \eqref{E:weakformMaxwell3} by taking $N \rightarrow +\infty$. (Here notice that since $\supp_v f^N (t) \subset D_{C_v}$ for all $t \in [0, T]$ and all $N$, the inner integral $\int_{\mathbb{R}^2} \hat{v}_2 f^N \varphi_3 dv$ in the last term above can be replaced by $\int_{D_{C_v}} \hat{v}_2 f^N \varphi_3 dv$.) Similarly we obtain \eqref{E:weakformMaxwell4} from \eqref{weakformMaxwellN4} by taking $N \rightarrow +\infty$. The proof of Theorem \ref{mtheorem} is complete.


\end{proof}

We now prove Lemma \ref{(extra)-converge}.

\begin{proof}
{\it{(of Lemma \ref{(extra)-converge}) }}

It suffices to consider the boundary point $x=0$ and show that
\begin{equation}
\label{E:extra21}
\tilde{\Xi}_N  :=  \int_0^T  \int_0^{1/N} \int_{\mathbb{R}^2} B_{ext, N} (x) f^N (t, x, v_1, v_2) \tilde{\alpha} (t, x, v_1, v_2) dv_2 dv_1 dx dt \rightarrow 0 \text{ as } N \rightarrow +\infty ,
\end{equation}
since the part corresponding to the $x=1$ boundary is similar. We first observe, by using the definition of $B_{ext,N}$, the change of variable $y=Nx$ and the fact that $f^N (t, x, v_1, v_2) = f_0 ((X_N, V_N) (0; t, x, v_1, v_2))$:
\begin{equation} \label{E:extra21-eq0}
\begin{split}
\tilde{\Xi}_N  
& = \int_0^T \int_{x \in (0, 1/N]} \int_{v_1} \int_{v_2} B_{ext, N} (x) f^N (t, x, v_1, v_2) \tilde{\alpha} (t, x, v_1, v_2) dv_2 dv_1 dx dt \\
& = \int_0^T  \int_{ y \in (0, 1]} \int_{v_1} \int_{v_2}  \Psi' (y) f_0 ((X_N, V_N) (0; t, N^{-1} y, v_1, v_2)) \tilde{\alpha} (t, N^{-1} y, v_1, v_2) dv_2 dv_1 dy dt . \\
\end{split}
\end{equation}
Let $\epsilon = \frac{1}{N^{1/10}} \in (0, 1)$, we have
\begin{equation} \label{E:extra21-eq0-1}
\begin{split}
\tilde{\Xi}_N  
& = \int_0^T  \int_{ y \in (0, 1 -\epsilon]} \int_{v_1} \int_{v_2}  \Psi' (y) f_0 ((X_N, V_N) (0; t, N^{-1} y, v_1, v_2)) \tilde{\alpha} (t, N^{-1} y, v_1, v_2) dv_2 dv_1 dy dt \\
& \quad + \int_0^T  \int_{ y \in ( 1 -\epsilon, 1]} \int_{v_1} \int_{v_2}  \Psi' (y) f_0 ((X_N, V_N) (0; t, N^{-1} y, v_1, v_2)) \tilde{\alpha} (t, N^{-1} y, v_1, v_2) dv_2 dv_1 dy dt   \\
& := \Xi_N  + \Xi_N' . \\
\end{split}
\end{equation}
We have
\begin{equation} \label{E:extra21-eq0-2}
\begin{split}
| \Xi_N' |  
& \lesssim  \epsilon =  \frac{1}{N^{1/10}} . \\
\end{split}
\end{equation}
Let us consider the following quantity 
\begin{equation} \label{E:extra21-eq1}
\begin{split}
\Xi_N  
& = \int_0^T  \int_{ y \in (0, 1 -\epsilon]} \int_{v_1} \int_{v_2}  \Psi' (y) f_0 ((X_N, V_N) (0; t, N^{-1} y, v_1, v_2)) \tilde{\alpha} (t, N^{-1} y, v_1, v_2) dv_2 dv_1 dy dt . \\
\end{split}
\end{equation}
Using the change of variable $ t \mapsto t^*$ (noticing that $t = t(t^*, y, v_1, v_2)$), 
the fact that $|J_N| =1$, as well as \eqref{t1t2mapping-real-error-eq2}, we have
\begin{equation}  
\begin{split}
\Xi_N
& =    \int_{ y \in (0, 1-\epsilon]} \int_{v_1} \int_{v_2} \int_{\{t^*: t \in [0, T]\}}  \Psi' (y) f_0 ((X_N, V_N) (0; t^*, N^{-1} \tilde{y}, -\tilde{v}_1, \tilde{v}_2))  \\
& \quad \cdot \tilde{\alpha} (t, N^{-1} y, v_1, v_2) dt^* dv_2 dv_1 dy . \\
\end{split}
\end{equation}
Changing $t$ to $t^*$, $y$ to $\tilde{y}$, $v_1$ to $\tilde{v}_1$,$v_2$ to $\tilde{v}_2$ in the integrand above, we write 
\begin{equation} \label{E:extra21-eq11}
\begin{split}
\Xi_N
& =    \int_{ y \in (0, 1-\epsilon]} \int_{v_1} \int_{v_2} \int_{\{t^*: t \in [0, T]\}}  \Psi' (\tilde{y}) f_0 ((X_N, V_N) (0; t^*, N^{-1} \tilde{y}, -\tilde{v}_1, \tilde{v}_2)) \\
& \quad \cdot \tilde{\alpha} ( t^*, N^{-1} \tilde{y}, \tilde{v}_1, \tilde{v}_2) dt^* dv_2 dv_1 dy + R_{1,N} , \\
\end{split}
\end{equation}
where
\begin{equation}  \label{E:extra21-R1N-def} 
\begin{split}
 R_{1,N} 
 &  :=  \int_{ y \in (0, 1-\epsilon]} \int_{v_1} \int_{v_2} \int_{\{t^*: t \in [0, T]\}}    f_0 ((X_N, V_N) (0; t^*, N^{-1} \tilde{y}, -\tilde{v}_1, \tilde{v}_2)) \\
& \quad \cdot \big\{ \Psi' (y)  \tilde{\alpha} (t, N^{-1} y, v_1, v_2) -  \Psi' (\tilde{y}) \tilde{\alpha} ( t^*, N^{-1} \tilde{y}, \tilde{v}_1, \tilde{v}_2) \big\} dt^* dv_2 dv_1 dy   \\
\end{split}
\end{equation}
with $\tilde{y} = N \tilde{x}$. 
Making the change of variables $y \mapsto x = N^{-1} y $ as well as $(x, v_1, v_2) \mapsto (\tilde{x}, \tilde{v}_1, \tilde{v}_2)$ in the right hand side of \eqref{E:extra21-eq11}, and applying Lemma \ref{t1t2mapping-real-error} concerning trajectories, we obtain
\begin{equation}  
\begin{split}
\Xi_N  
& = \int_{ x \in (0, 1/N-\epsilon/N]} \int_{v_1} \int_{v_2} \int_{\{t^*: t \in [0, T]\}}  N \Psi' (N\tilde{x}) f_0 ((X_N, V_N) (0; t^*, N^{-1} \tilde{y}, -\tilde{v}_1, \tilde{v}_2)) \\
& \quad \cdot \tilde{\alpha} ( t^*, N^{-1} \tilde{y}, \tilde{v}_1, \tilde{v}_2) dt^* dv_2 dv_1 dx + R_{1,N} \\
& = \int_{\{(t^*, \tilde{x}, \tilde{v}_1, \tilde{v}_2) : x \in (0, 1/N-\epsilon/N], v_1 \in \mathbb{R}, v_2 \in \mathbb{R}, t \in [0, T]\}} N \Psi' (N\tilde{x}) f_0 ((X_N, V_N) (0; t^*,  \tilde{x}, -\tilde{v}_1, \tilde{v}_2))  \\
& \quad \cdot \tilde{\alpha} (t^*,   \tilde{x}, \tilde{v}_1, \tilde{v}_2) |\mathcal{J}_N| d t^* d\tilde{v}_2 d\tilde{v}_1 d\tilde{x}  +  R_{1,N} .  \\
\end{split}
\end{equation}
Removing the Jacobian $\mathcal{J}_N$ and then making the change of variable $\tilde{x} \mapsto \tilde{y} = N \tilde{x}$ again, we obtain
\begin{equation} \label{E:extra21-eq8}
\begin{split}
\Xi_N
& =  \int_{\{(t^*, \tilde{x}, \tilde{v}_1, \tilde{v}_2) : x \in (0, 1/N-\epsilon/N ], v_1 \in \mathbb{R}, v_2 \in \mathbb{R}, t \in [0, T]\}}   N \Psi' (N\tilde{x}) f_0 ((X_N, V_N) (0; t^*,  \tilde{x}, -\tilde{v}_1, \tilde{v}_2))  \\
& \quad \cdot \tilde{\alpha} (t^*, \tilde{x}, \tilde{v}_1, \tilde{v}_2) d t^* d\tilde{v}_2 d\tilde{v}_1 d\tilde{x} + R_{2, N}  +  R_{1,N}   \\
& =  \int_{\{(t^*, \tilde{y}, \tilde{v}_1, \tilde{v}_2) : y \in (0, 1- \epsilon], v_1 \in \mathbb{R}, v_2 \in \mathbb{R}, t \in [0, T]\}}   \Psi' (\tilde{y}) f_0 ((X_N, V_N) (0; t^*, N^{-1} \tilde{y}, -\tilde{v}_1, \tilde{v}_2))  \\
& \quad \cdot \tilde{\alpha} (t^*, N^{-1} \tilde{y}, \tilde{v}_1, \tilde{v}_2) d t^* d\tilde{v}_2 d\tilde{v}_1 d\tilde{y} + R_{2, N}  +  R_{1,N}   , \\
\end{split}
\end{equation}
where 
\begin{equation}  \label{E:extra21-R2N-def}
\begin{split}
R_{2, N}
& :=  \int_{\{(t^*, \tilde{x}, \tilde{v}_1, \tilde{v}_2) : x \in (0, 1/N - \epsilon/N], v_1 \in \mathbb{R}, v_2 \in \mathbb{R}, t \in [0, T]\}}  N\Psi' (N\tilde{x}) f_0 ((X_N, V_N) (0; t^*, \tilde{x}, -\tilde{v}_1, \tilde{v}_2))  \\
& \quad \cdot \tilde{\alpha} (t^*, \tilde{x}, \tilde{v}_1, \tilde{v}_2)  \big( |\mathcal{J}_N| -1 \big) d t^* d\tilde{v}_2 d\tilde{v}_1 d\tilde{x}  .  \\
\end{split}
\end{equation}
We continue to rewrite \eqref{E:extra21-eq8} by specifying region of integration as follows: (Note that when $y$ runs through $(0, 1 - \epsilon]$, $\tilde{y} $ stays in $(0, 1 - \epsilon + O(1/N \epsilon)]$. However, we can constrain the $\tilde{y}$-integral on $(0, 1]$ since $\supp \Psi' \subset (0, 1]$. )
\begin{equation}  
\begin{split}
\Xi_N  
& = \int_{ \tilde{y} \in (0, 1- \epsilon + O(1/N \epsilon)]} \int_{\tilde{v}_1} \int_{\tilde{v}_2} \int_{\{t^*: t \in [0, T]\}}  \Psi' (\tilde{y}) f_0 ((X_N, V_N) (0; t^*, N^{-1} \tilde{y}, -\tilde{v}_1, \tilde{v}_2))  \\
& \quad \cdot \tilde{\alpha} (t^*, N^{-1} \tilde{y}, \tilde{v}_1, \tilde{v}_2) dt^* d\tilde{v}_2 d\tilde{v}_1 d\tilde{y} + R_{2, N}  +  R_{1,N}  \\
& = \int_{ \tilde{y} \in (0, 1- \epsilon  ]} \int_{\tilde{v}_1} \int_{\tilde{v}_2} \int_{\{t^*: t \in [0, T]\}}  \Psi' (\tilde{y}) f_0 ((X_N, V_N) (0; t^*, N^{-1} \tilde{y}, -\tilde{v}_1, \tilde{v}_2))  \\
& \quad \cdot \tilde{\alpha} (t^*, N^{-1} \tilde{y}, \tilde{v}_1, \tilde{v}_2) dt^* d\tilde{v}_2 d\tilde{v}_1 d\tilde{y} + R_{2, N}  +  R_{1,N} + O(\frac{1}{N \epsilon}) . \\
\end{split}
\end{equation}
Changing the name of the variables from $(\tilde{y}, \tilde{v}_1, \tilde{v}_2)$ to $(y, v_1, v_2)$ gives:
\begin{equation} \label{E:extra21-eq9}
\begin{split}
\Xi_N  
& = \int_{ y \in (0, 1- \epsilon]} \int_{v_1} \int_{v_2} \int_{\{t^*: t \in [0, T]\}}  \Psi' (y) f_0 ((X_N, V_N) (0; t^*, N^{-1} y, -v_1, v_2))  \\
& \quad \cdot \tilde{\alpha} (t^*, N^{-1} y, v_1, v_2) dt^* dv_2 dv_1 dy + R_{2, N}  +  R_{1,N} + O(\frac{1}{N \epsilon}) . \\
\end{split}
\end{equation}
Making a change of variable $v_1 \mapsto -v_1$ gives
\begin{equation} \label{E:extra21-eq2}
\begin{split}
 \Xi_N  
& =    \int_{ y \in (0, 1 - \epsilon]} \int_{v_1} \int_{v_2} \int_{\{t^*: t \in [0, T]\}}  \Psi' (y) f_0 ((X_N, V_N) (0; t^*, N^{-1} y, v_1, v_2)) \\
& \quad \cdot \tilde{\alpha} (t^*, N^{-1} y, -v_1, v_2) dt^* dv_2 dv_1 dy  + R_{2, N}  +  R_{1,N} + O(\frac{1}{N \epsilon}) . \\
\end{split}
\end{equation}
We want to change the region of integration in \eqref{E:extra21-eq2} to $ \{ (t^*, y, v_1, v_2) :  t^* \in [0, T], \ y \in (0, 1], \ v_1 \in \mathbb{R} , \ v_2 \in \mathbb{R} \}$. For this purpose, let 
\begin{equation}  \label{E:extra21-R3N-def}
\begin{split}
R_{3,N} 
&  :=  \int_0^{1- \epsilon} \int_{v_1} \int_{v_2} \int_{t^* \in \tilde{I}_{1,N} }  \Psi' (y) f_0 ((X_N, V_N)(0; t^*, N^{-1} y, v_1, v_2)) \tilde{\alpha} (t^*, N^{-1} y, -v_1, v_2) dt^* dv_2 dv_1 dy   \\
& \qquad - \int_0^{1- \epsilon} \int_{v_1} \int_{v_2} \int_{t^* \in \tilde{I}_{2,N} }  \Psi' (y) f_0 ((X_N, V_N)(0; t^*, N^{-1} y, v_1, v_2)) \tilde{\alpha} (t^*, N^{-1} y, -v_1, v_2) dt^* dv_2 dv_1 dy   \\
\end{split}
\end{equation}
with $\tilde{I}_{1,N} = \tilde{I}_{1,N} (y, v_1, v_2) :=  {\{t^*: t \in [0, T]\}}\setminus [0, T] $, $\tilde{I}_{2,N} = \tilde{I}_{2,N} (y, v_1, v_2) :=  [0, T]\setminus {\{t^*: t \in [0, T]\}}  $. \eqref{E:extra21-eq2} then becomes
\begin{equation}  
\begin{split}
 \Xi_N =
&   \int_{t^* \in [0, T]}  \int_{ y \in (0, 1 - \epsilon]} \int_{v_1} \int_{v_2}  \Psi' (y) f_0 ((X_N, V_N) (0; t^*, N^{-1} y, v_1, v_2)) \\
& \cdot \tilde{\alpha} (t^*, N^{-1} y, -v_1, v_2) dv_2 dv_1 dy dt^*  + R_{3,N} + R_{2, N}  +  R_{1,N} + O(\frac{1}{N \epsilon}) . \\
\end{split}
\end{equation}
Changing the name of the variable $t^*$ to $t$ gives
\begin{equation} \label{E:extra21-eq5}
\begin{split}
 \Xi_N =
&   \int_0^T \int_{ y \in (0, 1 - \epsilon]} \int_{v_1} \int_{v_2}  \Psi' (y) f_0 ((X_N, V_N) (0; t, N^{-1} y, v_1, v_2)) \\
& \cdot \tilde{\alpha} (t, N^{-1} y, -v_1, v_2) dv_2 dv_1 dy dt  + R_{3,N} + R_{2, N}  +  R_{1,N} + O(\frac{1}{N \epsilon}). \\
\end{split}
\end{equation}
Now adding \eqref{E:extra21-eq1} and \eqref{E:extra21-eq5} yields
\begin{equation*}
\begin{split}
2 \Xi_N =
& \int_0^T  \int_0^{1} \int_{v_1} \int_{v_2}  \Psi' (y) f_0 ((X_N, V_N)(0; t, N^{-1} y, v_1, v_2)) \\
&  \cdot \big[ \tilde{\alpha} (t, N^{-1} y, v_1, v_2) + \tilde{\alpha} (t, N^{-1} y, -v_1, v_2) \big] dv_2 dv_1 dy dt + R_{3,N} + R_{2, N}  +  R_{1,N} . \\
\end{split}
\end{equation*}
Applying the mean value theorem to $\tilde{\alpha}$ gives 
\begin{equation}  \label{E:extra21-eq7} 
\begin{split}
2 \Xi_N  
& = \int_0^T  \int_0^{1- \epsilon} \int_{v_1} \int_{v_2} \Psi' (y) f_0 ((X_N, V_N)(0; t, N^{-1} y, v_1, v_2)) \\
& \cdot \big[ \tilde{\alpha} (t, 0, v_1, v_2) + \partial_x \tilde{\alpha}  (t, N^{-1} \overline{y}^1, v_1, v_2) N^{-1} y \big] dv_2 dv_1 dy dt \\
& + \int_0^T  \int_0^{1} \int_{v_1} \int_{v_2}   \Psi' (y) f_0 ((X_N, V_N) (0; t^*, N^{-1} y, v_1, v_2)) \\
& \cdot \big[ \tilde{\alpha} (t^*, 0, -v_1, v_2) + \partial_x \tilde{\alpha} (t^*, N^{-1} \overline{y}^2, -v_1, v_2) N^{-1}  y \big] dv_2 dv_1 dy dt^*  \\
& + R_{3,N} + R_{2, N}  +  R_{1,N} + O(\frac{1}{N \epsilon}) .  \\
\end{split}
\end{equation}
Here $\overline{y}^1$ and $\overline{y}^2$ are between $0$ and $y$. Let 
\begin{equation} \label{E:extra21-R4N-def}
\begin{split}
R_{4,N}  
& := \int_0^T  \int_0^{1- \epsilon} \int_{v_1} \int_{v_2} \Psi' (y) f_0 ((X_N, V_N)(0; t, N^{-1} y, v_1, v_2))   \\
& \quad  \cdot  \partial_x \tilde{\alpha}  (t, N^{-1} \overline{y}^1 , v_1, v_2) N^{-1}  y dv_2 dv_1 dy dt   \\
& \quad  +  \int_0^T  \int_0^{1- \epsilon} \int_{v_1} \int_{v_2}   \Psi' (y) f_0 ((X_N, V_N) (0; t^*, N^{-1} y, v_1, v_2))   \\
& \quad  \cdot  \partial_x \tilde{\alpha} (t^*, N^{-1} \overline{y}^2, - v_1, v_2) N^{-1}  y  dv_2 dv_1 dy dt^* ,  \\
\end{split}
\end{equation}
then we have \eqref{E:extra21-eq7} becomes
\begin{equation} \label{E:extra21-eq6}
\begin{split}
2 \Xi_N  
& = \int_0^T  \int_0^{1- \epsilon} \int_{v_1} \int_{v_2} \Psi' (y) f_0 ((X_N, V_N)(0; t, N^{-1} y, v_1, v_2))   \tilde{\alpha} (t, 0, v_1, v_2)   dv_2 dv_1 dy dt \\
& + \int_0^T  \int_0^{1 - \epsilon} \int_{v_1} \int_{v_2}   \Psi' (y) f_0 ((X_N, V_N) (0; t^*, N^{-1} y, v_1, v_2)) \tilde{\alpha} (t^*, 0, -v_1, v_2)  dv_2 dv_1 dy dt^* \\
& + R_{4,N} + R_{3,N} + R_{2, N}  +  R_{1,N}  + O(\frac{1}{N \epsilon})  \\
\end{split}
\end{equation}
We recall that $\tilde{\alpha} (t^*, 0, v_1, v_2) = -\tilde{\alpha} (t^*, 0, -v_1, v_2)$ for each $t^* \in [0, T)$, which causes the cancellation of the first two terms in \eqref{E:extra21-eq6}. This leads us to
\begin{equation} \label{E:extra21-eq10}
\begin{split}
2 \Xi_N  
& = R_{4,N} + R_{3,N} + R_{2, N}  +  R_{1,N} + O(\frac{1}{N \epsilon})   .  \\
\end{split}
\end{equation}

We now estimate the error terms $R_{1,N}$, $R_{2,N}$, $R_{3,N}$ and $R_{4,N}$, making use of Lemma \ref{NS3-lm3.1}, Corollary \ref{NS3-lm3.1-cor}, Lemma \ref{t1t2mapping-model} and Lemma \ref{t1t2mapping-real-error}.

\begin{flushleft}
{\it{Estimate of $R_{1,N}$ (defined in \eqref{E:extra21-R1N-def}): }}  
\end{flushleft}
From Lemma \ref{t1t2mapping-model} and Lemma \ref{t1t2mapping-real-error}, we have $|t^* - t| = O(\frac{1}{N \epsilon})$, $ |\tilde{y} - y| = O(\frac{1}{N \epsilon}) $, $ |\tilde{v}_1 - v_1| = O(\frac{1}{N \epsilon}) $, $ |\tilde{v}_2 - v_2| = O(\frac{1}{N \epsilon}) $. Notice that $\|\nabla_{t, v_1, v_2} \tilde{\alpha} \|_{L^\infty} \lesssim 1$, $\|\partial_y \tilde{\alpha} ( t, N^{-1} y, -v_1, v_2)\|_{L^\infty} \lesssim 1/N$ since $\tilde{\alpha}$ is a test function. Due to Corollary \ref{NS3-lm3.1-cor}, there exists $y_0 >0$ (depends on $f_0$) independent of $N$ and small enough such that $\supp_x f^N \subset (N^{-1} y_0, 1-N^{-1} y_0)$. Hence the bounds for the integral on $y$ can be replaced by $\int^{1- \epsilon}_{y=y_0}$, and $|\Psi'(y)| \leq |\Psi'(y_0)|$, $|\Psi''(y)| \leq |\Psi''(y_0)|$, since for any $t \in [0, T]$ and $(v_1, v_2) \in \mathbb{R}^2$,  $N^{-1}y \in (N^{-1} y_0, 1-N^{-1} y_0)$ if $f^N (t, N^{-1} y, v_1, v_2) \neq 0$. Notice that $N^{-1} \tilde{y} =\tilde{x} = X_N (t^*;t, x, v_1, v_2) \in \supp_x f^N \subset (N^{-1} y_0, 1-N^{-1} y_0)$, so $|\Psi'(\tilde{y})| \leq |\Psi'(y_0)|$, $|\Psi''(\tilde{y})| \leq |\Psi''(y_0)|$ hold too. Using these facts, we deduce
\begin{equation}
\begin{split}
& \quad |\Psi' (y)  \tilde{\alpha} (t, N^{-1} y, v_1, v_2) -  \Psi' (\tilde{y}) \tilde{\alpha} ( t^*, N^{-1} \tilde{y}, \tilde{v}_1, \tilde{v}_2)| \\
& \leq |\Psi' (y)  -  \Psi' (\tilde{y})| | \tilde{\alpha} (t, N^{-1} y, v_1, v_2)| +| \Psi' (\tilde{y}) | |\tilde{\alpha} (t, N^{-1} y, v_1, v_2)- \tilde{\alpha} ( t^*, N^{-1} \tilde{y}, \tilde{v}_1, \tilde{v}_2)| \\
& \lesssim |\tilde{y}-y| + |t^*-t| + |\tilde{v}_1- v_1| + |\tilde{v}_2- v_2| \\
& \lesssim \frac{1}{N \epsilon} , \\
\end{split}
\end{equation}
from which we obtain 
\begin{equation} \label{E:extra21-R1N-est}
|R_{1, N}| \lesssim \frac{1}{N \epsilon} .
\end{equation}

\begin{flushleft}
{\it{Estimate of $R_{2,N}$ (defined in \eqref{E:extra21-R2N-def}): }}  
\end{flushleft}
Notice that $\big| |\mathcal{J}_N| -1 \big| \lesssim \frac{1}{N \epsilon^3} $ for any $t^*\in [0, T]$ (by Lemma \ref{t1t2mapping-real-error}). Moreover, again due to Corollary \ref{NS3-lm3.1-cor},  
for any $t \in [0, T]$ and $(v_1, v_2) \in \mathbb{R}^2$, $N^{-1} \tilde{y} =\tilde{x}  \in (N^{-1} y_0, 1-N^{-1} y_0)$ if $f^N (t, N^{-1} \tilde{y}, v_1, v_2) \neq 0$. Hence for the integral $R_{2, N}$ we have $|\Psi' (\tilde{y}) |\leq |\Psi' (y_0)|$. We deduce, using $|\tilde{y} - y| \lesssim \frac{1}{N \epsilon}$: 
\begin{equation}   \label{E:extra21-R2N-est}
\begin{split}
|R_{2, N}|
& \leq  \int_{\{(t^*, \tilde{x}, \tilde{v}_1, \tilde{v}_2) : y \in (0, 1-\epsilon], v_1 \in \mathbb{R}, v_2 \in \mathbb{R}, t \in [0, T]\}}    | \Psi' (\tilde{y}) |   f_0 ((X_N, V_N) (0; t^*, N^{-1} \tilde{y}, -\tilde{v}_1, \tilde{v}_2))  \\
& \quad \cdot |\tilde{\alpha} (t^*, N^{-1} \tilde{y}, \tilde{v}_1, \tilde{v}_2)| d t^* d\tilde{v}_2 d\tilde{v}_1 d\tilde{y} \cdot \big\| |\mathcal{J}_N| -1 \big\|_{L^\infty}    \\
& \lesssim \frac{1}{N \epsilon^3} \int_{\{(t^*, \tilde{y}, \tilde{v}_1, \tilde{v}_2) : y \in (0, 1-\epsilon], v_1 \in \mathbb{R}, v_2 \in \mathbb{R}, t \in [0, T]\}}   |\Psi' (y_0)| f_0 ((X_N, V_N) (0; t^*, N^{-1} \tilde{y}, -\tilde{v}_1, \tilde{v}_2))  \\
& \quad \cdot  |\tilde{\alpha} (t^*, N^{-1} \tilde{y}, \tilde{v}_1, \tilde{v}_2)|  d t^* d\tilde{v}_2 d\tilde{v}_1 d\tilde{y}    \\
& \leq  \frac{1}{N \epsilon^3} \int_{\{(t^*, \tilde{y}, \tilde{v}_1, \tilde{v}_2) : \tilde{y} \in (0, 2], v_1 \in \mathbb{R}, v_2 \in \mathbb{R}, t \in [0, T]\}}  |\Psi' (y_0)| f_0 ((X_N, V_N) (0; t^*, N^{-1} \tilde{y}, -\tilde{v}_1, \tilde{v}_2))  \\
& \quad \cdot |\tilde{\alpha} (t^*, N^{-1} \tilde{y}, \tilde{v}_1, \tilde{v}_2)|  d t^* d\tilde{v}_2 d\tilde{v}_1 d\tilde{y}  \\
& \lesssim \frac{1}{N \epsilon^3} . \\
\end{split}
\end{equation}

\begin{flushleft}
{\it{Estimate of $R_{3,N}$ (defined in \eqref{E:extra21-R3N-def}): }}  
\end{flushleft}
Notice that for any fixed $(y, v_1, v_2)$, the mapping $t \mapsto t^*$ is a translation, where the amount of the translation is a function of $(y, v_1, v_2)$ and is of $O(\frac{1}{N \epsilon})$. We have $|\tilde{I}_{j,N}| \lesssim \frac{1}{N \epsilon}$ for each $(y, v_1, v_2)$ and $j=1, \ 2$. Also, again due to Corollary \ref{NS3-lm3.1-cor}, $\supp_x f^N \subset (N^{-1} y_0, 1-N^{-1} y_0)$, and hence for the integral $R_{3, N}$ we have $|\Psi' (y) |\leq |\Psi' (y_0)|$. From this we deduce 
\begin{equation} \label{E:extra21-R3N-est}
\begin{split}
& \quad |R_{3,N} |  \\
& \leq | \int_{y_0}^{1-\epsilon} \int_{v_1} \int_{v_2} \int_{t^* \in \tilde{I}_{1,N} }  \Psi' (y) f_0 ((X_N, V_N)(0; t^*, N^{-1} y, v_1, v_2)) \tilde{\alpha} (t^*, N^{-1} y, -v_1, v_2) dt^* dv_2 dv_1 dy |  \\
& \qquad + | \int_{y_0}^{1-\epsilon} \int_{v_1} \int_{v_2} \int_{t^* \in \tilde{I}_{2,N} }  \Psi' (y) f_0 ((X_N, V_N)(0; t^*, N^{-1} y, v_1, v_2)) \tilde{\alpha} (t^*, N^{-1} y, -v_1, v_2) dt^* dv_2 dv_1 dy |  \\
& \lesssim \int_{y_0}^{1-\epsilon} \int_{v_1} \int_{v_2} \int_{t^* \in \tilde{I}_{1,N} }  |\Psi' (y_0)| |f_0 ((X_N, V_N)(0; t^*, N^{-1} y, v_1, v_2))| |\tilde{\alpha} (t^*, N^{-1} y, -v_1, v_2)| dt^* dv_2 dv_1 dy   \\
& \qquad + \int_{y_0}^{1-\epsilon} \int_{v_1} \int_{v_2} \int_{t^* \in \tilde{I}_{2,N} }  |\Psi' (y_0)| |f_0 ((X_N, V_N)(0; t^*, N^{-1} y, v_1, v_2))| |\tilde{\alpha} (t^*, N^{-1} y, -v_1, v_2)| dt^* dv_2 dv_1 dy   \\
& \lesssim \frac{1}{N \epsilon} .  \\
\end{split}
\end{equation}

\begin{flushleft}
{\it{Estimate of $R_{4,N}$ (defined in \eqref{E:extra21-R4N-def}): }}  
\end{flushleft}
Noticing that $\| \partial_x \tilde{\alpha} \|_{L^\infty } \lesssim 1$, and by Corollary \ref{NS3-lm3.1-cor} the bounds for the $y$-integration can be replaced by $\int_{y=y_0}^{1}$, we can estimate the terms involving $\partial_x \tilde{\alpha}$ in \eqref{E:extra21-eq6} as 
\begin{equation} \label{E:extra21-R4N-est}
\begin{split}
|R_{4,N} |
& \leq  |\Psi' (y_0)| \Big\{  \int_0^T  \int_{y_0}^{1-\epsilon} \int_{v_1} \int_{v_2}  |f_0 ((X_N, V_N)(0; t, N^{-1} y, v_1, v_2))| \\
& \quad  \cdot  | \partial_x \tilde{\alpha}  (t, N^{-1} \overline{y}^1, v_1, v_2) | N^{-1} y  dv_2 dv_1 dy dt  \\
& \quad + \int_0^T  \int_{y_0}^{1-\epsilon} \int_{v_1} \int_{v_2}  | f_0 ((X_N, V_N) (0; t^*, N^{-1} y, v_1, v_2)) | \\
& \quad  \cdot  | \partial_x \tilde{\alpha} (t^*, N^{-1} \overline{y}^2 , -v_1, v_2) | N^{-1}  y dv_2 dv_1 dy dt^* \Big\} \\
& \lesssim \frac{1}{N}  . \\
\end{split}
\end{equation}
Plugging the estimates for the error terms \eqref{E:extra21-R1N-est}, \eqref{E:extra21-R2N-est}, \eqref{E:extra21-R3N-est}, \eqref{E:extra21-R4N-est} into \eqref{E:extra21-eq10}, we have
\begin{equation*}
| 2 \Xi_N | \lesssim \frac{1}{N \epsilon^3} =  \frac{1}{N^{7/10}}  .
\end{equation*}
Combining this together with \eqref{E:extra21-eq0-2} and recalling \eqref{E:extra21-eq0-1} immediately gives \eqref{E:extra21}, and therefore \eqref{E:extra20} is proved.
\end{proof}

\section{A Model with Finite Magnetic Confinement} \label{SectionFinite}

The external magnetic field $B_{ext, N}$ given in Section \ref{SectionSetup} can be replaced by a \emph{finite} version, which is physically more reasonable: Let $N \geq 8$, and $b(x)$ be a $C^1$, piecewise $C^3$, compactly supported function on the closed half-line $[0, +\infty)$ that satisfies
\begin{equation} \label{b-prop-finite}
\begin{split}
&   b' (x) >0  , \ b'' (x) <0 , \text{ and } b''' (x) >0 \text{ on } [0, 1) , \\
& b (x) = 0 \ \text{when} \ x \in (1, +\infty) . \\
\end{split}
\end{equation} 
(Notice that in \eqref{b-prop-finite} the function $b(x)$ no longer blows up to $\infty$ when $x \rightarrow 0$. ) We define $B_{ext,N}$, $\psi_{ext,N}$ and $\Psi$ as in \eqref{BextN-b-rescale}, \eqref{psiextN-potentialdef} and \eqref{Psi-potentialdef}, respectively. 
Now the function $\Psi$ is a $C^2$, piecewise $C^4$, bounded and compactly supported function on $[0, +\infty)$ that satisfies
\begin{equation} \label{Psi-prop-finite}
\begin{split}
& \Psi (x) >0,  \ \Psi'' (x) >0 , \ \Psi''' (x) <0 ,  \ \Psi'''' (x) >0 \text{ on } [0, 1) , \\
& \Psi (x) = 0 \ \text{when} \ x \in (1, +\infty) . \\ 
\end{split}
\end{equation} 
Recall that $ B_{ext, N} (x) = \partial_x \psi_{ext, N} (x) $, and  
\begin{equation} \label{psiextN-def-finite}
\begin{split}
& \psi_{ext, N} (x) = \Psi(Nx) \text{ for } x \in [0, \frac{1}{2}], \\
& \psi_{ext, N} (x) = \Psi(N(1-x)) \text{ for } x \in [\frac{1}{2}, 1] . \\
\end{split}
\end{equation}
We take the initial-boundary data as described in Section \ref{SectionSetup}. Again, without loss of generality, we assume $N$ is sufficiently large such that \eqref{cond-N} holds.

We use $(f^N, E^N_1, E^N_2, B^N)$ to denote the solutions for this 1.5D RVM with this \emph{finite} external magnetic confinement, assuming that these solutions exist. It turns out that if $\Psi$ is chosen to be large enough (though finite), then the plasma is still confined away from the boundary, as stated in the following lemma, which is parallel to Lemma \ref{NS3-lm3.1}:

\begin{lemma} \label{NS3-lm3.1-finite}
Suppose $\supp_{x, v} f_0 (x, v) \subset [\epsilon_0, 1-\epsilon_0] \times \{|v| \leq k_0 \} $. 
Denote $P_N (t) := \sup \{ |v|: f^N (t, x, v) \neq 0 \text{ for some } x \in \Omega \}$. We have \\
1) 
\begin{equation}  \label{NS3-lm3.1-PN-finite}
P_N (t) \leq C_v  := k_0 + 2 C_1 T , \text{ for all } t \in  [0, 2T]  . 
\end{equation}
Hence the support of $f^N$ in $v$ is contained in the disk $\overline{D}_{C_v}$. \\ 
2) When $N \geq \epsilon_0^{-1}$,
\begin{equation} \label{NS3-lm3.1-psiextN-finite}
\|\psi_{ext, N} \|_{L^\infty (\supp_x f^N )} \leq C_2 .
\end{equation}
where
$ C_2  :=   2 k_0 + 4 C_1 T + 2 C_1  $ is as defined in Lemma \ref{NS3-lm3.1}. \\
3) If $\Psi$ is taken such that $\Psi (0) > C_2$ (Notice that $C_2$ does not depend on $\Psi$), then the support of $f^N$ in $x$ stays away from the boundary $\partial \Omega$ with a positive distance no less than $\Psi^{-1} (C_2)$, i.e. $dist ( \supp_x f^N, \partial \Omega ) \geq \Psi^{-1} (C_2) > 0$ on $[0, 2T]$. 
\end{lemma}
\begin{proof}
The proof for Lemma \ref{NS3-lm3.1-finite} is very similar to the one for Lemma \ref{NS3-lm3.1}. 
The ODE for the particle trajectory is
\begin{equation}
\left\{
\begin{aligned}
&\dot X_N = \hat V_{1,N} \\
&\dot V_{1,N} =  E^N_1 (t, X_N) + \hat V_{2,N} B^N (t, X_N) + \hat V_{2,N} B_{ext, N} (X_N) \\
&\dot V_{2,N} =  E^N_2 (t, X_N) - \hat V_{1,N} B^N (t, X_N) - \hat V_{1,N} B_{ext, N} (X_N)
\end{aligned}
\right.
\end{equation}
with initial data $X_N(0) = x$, $V_{1,N} (0) = v_1$, $V_{2,N} (0) = v_2$. 
The proofs for 1) and 2) are exactly the same as the proof for 1) and 2) in Lemma \ref{NS3-lm3.1} so we omit them. 

We now prove 3). Assume $\Psi (0) > C_2$ (Notice that $C_2$ does not depend on $\Psi$). Recall the monotonicity assumption in \eqref{Psi-prop-finite}, we have $dist ( \supp_x f^N, \partial \Omega ) \geq \Psi^{-1} (C_2) > 0$ on $[0, 2T]$ (Notice that $\Psi (0) > C_2 > 0$ together with \eqref{Psi-prop-finite} implies that $ \Psi^{-1} (C_2)$ exists and is in $(0, 1)$). Hence the support of $f^N$ in $x$ stays away from the boundary $\partial \Omega$ with a positive distance, i.e. $dist ( \supp_x f^N, \partial \Omega ) \geq \Psi^{-1} (C_2) > 0$ on $[0, 2T]$. The proof of the lemma is complete.
\end{proof}
 
Moreover, similarly as Corollary \ref{NS3-lm3.1-cor}, we also have  
\begin{corollary} \label{NS3-lm3.1-cor-finite}
There exists $y_0 >0$ (depends on $f_0$) independent of $N$ and small enough such that $\supp_x f^N \subset (N^{-1} y_0, 1-N^{-1} y_0)$ for all $t \in [0, 2T]$. For any $x \in \supp_x f^N$, we have $|\Psi (Nx) |\leq |\Psi (y_0)|$, $|\Psi' (Nx) |\leq |\Psi' (y_0)|$, $|\Psi'' (Nx) |\leq |\Psi'' (y_0)|$, $|\Psi''' (Nx) |\leq |\Psi''' (y_0)|$. 
\end{corollary}

\begin{proof}
Let $y_0 :=  \Psi^{-1} (C_2)$. The proof for this corollary is essentially the same as the one for Corollary \ref{NS3-lm3.1-cor} so we omit it.
\end{proof}

Due to Lemma \ref{NS3-lm3.1-finite}, we learn that if $\Psi$ is chosen to satisfy \eqref{Psi-prop-finite} and $\Psi (0) > C_2 $, then no boundary condition on $f^N$ is needed for the 1.5D RVM model. Following the proof in \cite{NS3}, we can establish the global well-posedness and $C^1$ regularity on $[0, 2T]$ (\emph{but not on any larger time interval}) for the solution $(f^N, E^N_1, E^N_2, B^N)$. By the same argument as in Section \ref{SectionWL} -- Section \ref{ProofofMainTheorem}, we obtain exactly the same result as Theorem \ref{mtheorem} for this \emph{finitely-confined} model. This is because that Theorem \ref{mtheorem} only concerns the behavior of the plasma when $N \rightarrow +\infty$.

\section{A Two-Species Model} \label{SectionTwoSpecies}

We can also consider the two-species 1.5D RVM system on a bounded interval $\Omega = (0, 1)$, with the same external magnetic field $ B_{ext, N}$ as described in Section \ref{Intro} and Section \ref{SectionSetup}. The Vlasov equation is now
\begin{equation}  \label{VlasovBext-TwoSpecies}
\partial_t f^\pm + \hat{v}_1 \partial_x f^\pm \pm ( E_1 + \hat{v}_2 B + \hat{v}_2 B_{ext, N} ) \partial_{v_1} f^\pm \pm ( E_2 - \hat{v}_1 B - \hat{v}_1 B_{ext, N} ) \partial_{v_2} f^\pm =0   
\end{equation} 
with $f^\pm (t, x, v)$ being the particle density function for the ions and electrons, respectively. The Maxwell equations remain the same form as \eqref{MaxwellBext}:
\begin{equation}  \label{MaxwellBext-TwoSpecies}
\begin{split}
& \partial_t E_1=  - j_1 \ , \ \partial_x E_1 = \rho \ , \\
& \partial_t E_2 = - \partial_x B - j_2 \ , \\
& \partial_t B = - \partial_x E_2 \ . \\ 
\end{split}
\end{equation} 
with $\rho (t, x) :=  \int_{\mathbb{R}^2} (f^+ (t, x, v) - f^- (t,x, v) ) dv$ and $\textbf{j} (t, x) :=  \int_{\mathbb{R}^2} \hat{v} (f^+ (t, x , v) - f^- (t,x, v) ) dv$. Note that we have normalized the speed of light as well as the unit mass and charge of the particles to be $1$ since these quantities play minor roles in the qualitative analysis, while in reality the ions are much heavier than the electrons.

Similar as in the one-species case, we put down the following initial-boundary conditions:
\begin{equation} \label{boundarycondBext-TwoSpecies}
\begin{split}
& 0 \leq f^\pm (0, x, v) = f^\pm_0 (x, v) \in C^1_0 (\Omega \times\mathbb{R}^2)   , \\
& \supp_{x, v} f^\pm_0 (x, v) \subset [\epsilon_0, 1-\epsilon_0] \times \{|v| \leq k_0 \}  , \\
&  E_1(0, x) =  \int^x_0 \int_{\mathbb{R}^2} ( f^+_0 (y, v) -f^-_0 (y, v) ) dv dy + \lambda   =: E_{1, 0} (x) \in C^1 , \  \ E_1(t, 0) \equiv \lambda ,    \\
& E_2(0, x) = E_{2, 0} (x) \in C^1 , \ B(0, x) = B_{0} (x)  \in C^1 , \\
& E_2 (t, 0) = E_{2, b} (t)  \in C^1  , \  B (t, 1)  = B_b (t)  \in C^1 , \\
\end{split}
\end{equation} 
where $E_{2, b} $, $E_{2, 0}$, $B_b$ and $B_0$ satisfy
\begin{equation}
E_{2,b} (0) = E_{2, 0} (0), \ B_b (0) = B_0 (1)  
\end{equation}
for the sake of compatibility.

It can be shown that the particles will not hit the boundary, due to the confining property of $B_{ext,N}$ (See Lemma \ref{NS3-lm3.1-TwoSpecies} below). Therefore no boundary condition on $f^\pm$ is needed for \eqref{VlasovBext-TwoSpecies}.


We also consider the two-species 1.5D RVM on $\Omega$ with no external magnetic field.
The Vlasov equation is
\begin{equation}  \label{VlasovBextspecular-TwoSpecies}
\partial_t f^\pm + \hat{v}_1 \partial_x f \pm ( E_1 + \hat{v}_2 B ) \partial_{v_1} f^\pm \pm ( E_2 - \hat{v}_1 B   ) \partial_{v_2} f^\pm =0  \ , 
\end{equation} 
and the Maxwell remain the same form as \eqref{MaxwellBextspecular}:
\begin{equation}  \label{MaxwellBextspecular-TwoSpecies}
\begin{split}
& \partial_t E_1=  - j_1 \ , \ \partial_x E_1 = \rho \ , \\
& \partial_t E_2 = - \partial_x B - j_2 \ , \\
& \partial_t B = - \partial_x E_2 \ . \\ 
\end{split}
\end{equation}
with $\rho (t, x) :=  \int_{\mathbb{R}^2} (f^+ (t, x, v) - f^- (t,x, v) ) dv$ and $\textbf{j} (t, x) :=  \int_{\mathbb{R}^2} \hat{v} (f^+ (t, x , v) - f^- (t,x, v) ) dv$. 
Again we put down the initial-boundary conditions \eqref{boundarycondBext-TwoSpecies} together with the specular boundary condition 
\begin{equation}  \label{specularBC-TwoSpecies}
f^\pm (t, x, v_1, v_2) = f^\pm (t, x, -v_1, v_2 ), \ \text{for} \ x = 0, \ 1 , 
\end{equation}



We define the weak solutions of the two-species RVM by analogous ways as in Definition \ref{D:weak-Bext} and Definition \ref{D:weak}. Without loss of generality we assume $N$ is sufficiently large such that \eqref{cond-N} is satisfied, with $\supp_x f_0 (x, v)$ in \eqref{cond-N} being replaced by $\cup_\pm \supp_x f_0^\pm (x, v)$.
 
The global well-posedness and $C^1$ regularity of the system \eqref{VlasovBext-TwoSpecies} and \eqref{MaxwellBext-TwoSpecies} with the conditions \eqref{boundarycondBext-TwoSpecies} can be given in essentially the same way as in \cite{NS3}. In fact, by essentially the same argument as in Section 2 in \cite{NS3} and in Section \ref{SectionWL} in this paper, we have
\begin{lemma} \label{lm3-TwoSpecies}
For any $T>0$, there exists a constant $\tilde{C} >0$ (which only depends on the initial-boundary data and $T$, in particular, independent of $N$), such that for all $N$ large enough such that \eqref{cond-N} holds 
\begin{equation}
\|(E^N_1, E^N_2, B^N) \|_{C^1 ([0, T] \times \Omega)} \leq \tilde{C} .  
\end{equation}
\end{lemma}

In particular, again we have the observation that the particles can not hit the boundary if their initial position are away from it, due to the confining property of $B_{ext,N}$. We state this fact in the lemma below, whose proof is essentially the same as the one for Lemma \ref{NS3-lm3.1}. Therefore no boundary condition on $f$ is needed for \eqref{VlasovBext-TwoSpecies}.

\begin{lemma} \label{NS3-lm3.1-TwoSpecies}
Suppose $\supp_{x, v} f_0^\pm (x, v) \subset [\epsilon_0, 1-\epsilon_0] \times \{|v| \leq k_0 \} $. 
Denote $P_N^\pm (t)  :=  \sup \{ |v|: f^{N, \pm} (t, x, v) \neq 0 \text{ for some } x \in \Omega \}$ and $P_N (t)  :=  \max_{\pm} \{P_N^\pm (t)\} $. We have: \\
1) 
\begin{equation}  \label{NS3-lm3.1-PN-TwoSpecies}
P_N (t) \leq C_v  :=  k_0 + C_1 T , \text{ for all } t \in [0, T]  ,
\end{equation}
where $C_1$, $C_v$ are as defined in Lemma \ref{NS3-lm3.1}. 
Hence the support of $f^{N, \pm}$ in $v$ are contained in the disk $\overline{D}_{C_v}$. \\
2) When $N \geq \epsilon_0^{-1}$,
\begin{equation} \label{NS3-lm3.1-psiextN-TwoSpecies}
\|\psi_{ext, N} \|_{L^\infty (\cup_{\pm} \supp_x f^{N, \pm} )} \leq C_2  ,
\end{equation}
where
$ C_2  :=   2 k_0 + 2 C_1 T + 2 C_1  $ is as defined in Lemma \ref{NS3-lm3.1}.

\end{lemma}
{\it{Remark.} The inequality \eqref{NS3-lm3.1-psiextN-TwoSpecies} tells us that the supports of $f^{N, \pm}$ in $x$ stay away from the boundary $\partial \Omega$ with a positive distance, i.e. $dist ( \cup_{\pm} \supp_x f^{N, \pm}, \partial \Omega ) > 0$ on $[0, T]$. }
\begin{proof}
For 1) and 2), it suffices to prove
\begin{equation}  \label{NS3-lm3.1-PN-electron}
P_N^- (t) \leq C_v  = k_0 + C_1  T , \text{ for all } t \in [0, T]  
\end{equation}
and 
\begin{equation} \label{NS3-lm3.1-psiextN-electron}
\|\psi_{ext, N} \|_{L^\infty (  \supp_x f^{N, -} )} \leq C_2  ,
\end{equation}
respectively. Then combining \eqref{NS3-lm3.1-PN-electron}, \eqref{NS3-lm3.1-psiextN-electron} together with Lemma \ref{NS3-lm3.1} gives the desired results. 




The proof for \eqref{NS3-lm3.1-PN-electron} is very similar to the one for \eqref{NS3-lm3.1-PN} and we omit it. 
The ODE for the particle trajectory of an \emph{electron} is 
\begin{equation}
\left\{
\begin{aligned}
&\dot X_N = \hat V_{1,N} \\
&\dot V_{1,N} =  -E^N_1 (t, X_N) - \hat V_{2,N} B^N (t, X_N) - \hat V_{2,N} B_{ext, N} (X_N) \\
&\dot V_{2,N} =  -E^N_2 (t, X_N) + \hat V_{1,N} B^N (t, X_N) + \hat V_{1,N} B_{ext, N} (X_N)
\end{aligned}
\right.
\end{equation}
with initial data $X_N(0) = x$, $V_{1,N} (0) = v_1$, $V_{2,N} (0) = v_2$.


The proof for \eqref{NS3-lm3.1-PN-electron} is exactly the same as the proof for \eqref{NS3-lm3.1-PN} so we omit it. Next, we let $\psi^N (\tau, y) = \int^y_{1/2} B^N(\tau, z) dz$. 
We define
\begin{equation*}
p(\tau, y, w) :=  w_2 -\psi^N (\tau,y) - \psi_{ext,N} (y) ,
\end{equation*}
where $w = (w_1, w_2) \in \mathbb{R}^2$. Differentiating $p(\tau, y, w)$ along the characteristics, we obtain
\begin{equation*}
\begin{split}
& \quad \frac{d}{ds} p(s, X_N (s), V_N (s)) \\
& = \dot{V}_{2,N} (s) - \partial_t \psi^N (s, X_N (s)) - \dot{X}_N \partial_x \psi^N (s, X_N (s)) - \dot{X}_N \partial_x \psi_{ext,N} ( X_N (s)) \\
& = - E^N_2 (s, X_N(s)) + \hat{V}_{1,N} (s) [B^N (s,X_N(s)) + B_{ext,N} ( X_N(s))] \\
& \quad - \partial_t \psi^N (s, X_N (s)) - \hat{V}_{1,N} (s) [B^N (s,X_N(s)) + B_{ext,N} ( X_N(s))] \\
& = - E^N_2 (s, X_N(s)) - \partial_t \psi^N (  X_N (s)) \\
& = - E^N_2 (s, \frac{1}{2}) . \\
\end{split}
\end{equation*}
Here we used the fact that $\partial_t B^N = -\partial_x E^N_2$.
Integrating yields
\begin{equation*}
V_{2,N}(s) - \psi^N(s, X_N(s))- \psi_{ext,N}(X_N(s)) = v_2 - \psi^N (0,x) - \psi_{ext,N} (x) - \int^s_0 E^N_2 (\tau, \frac{1}{2}) d\tau ,
\end{equation*}
and hence
\begin{equation*}
|\psi_{ext,N}(X_N(s))| \leq 
|V_{2,N}(s)| + |\psi^N(s, X_N(s))| + | v_2| + |\psi^N (0,x)| + |\psi_{ext,N} (x)| + \int^s_0 |E^N_2 (\tau, \frac{1}{2})| d\tau .
\end{equation*}
Combining this with \eqref{NS3-Section2} and \eqref{NS3-lm3.1-PN} and noticing that when $N \geq \epsilon_0^{-1}$, $ \|\psi_{ext, N} \|_{L^\infty ([\epsilon_0, 1-\epsilon_0])} =0$, we have
\begin{equation*}
|\psi_{ext,N}(X_N(s))| \leq C_2 , \text{ for all } s \in [0, T] .
\end{equation*}
This inequality holds for all the trajectories. Hence we conclude \eqref{NS3-lm3.1-psiextN-electron}.
\end{proof}

The analogous result of Corollary \ref{NS3-lm3.1-cor} also holds by the same argument:

\begin{corollary} \label{NS3-lm3.1-cor-TwoSpecies}
There exists $y_0 >0$ (depends on $f_0$) independent of $N$ and small enough such that $\cup_\pm \supp_x f^{N, \pm} \subset (N^{-1} y_0, 1-N^{-1} y_0)$. For any $x \in \supp_x f^N$, we have $|\Psi (Nx) |\leq |\Psi (y_0)|$, $|\Psi' (Nx) |\leq |\Psi' (y_0)|$, $|\Psi'' (Nx) |\leq |\Psi'' (y_0)|$, $|\Psi''' (Nx) |\leq |\Psi''' (y_0)|$. 
\end{corollary}


 
Let $(f^{N, \pm}, E^N_1, E^N_2, B^N)$ be the (global-in-time) solution of the system \eqref{VlasovBext-TwoSpecies} and \eqref{MaxwellBext-TwoSpecies}, with the conditions \eqref{boundarycondBext-TwoSpecies}. We want to obtain a result analogous to Theorem \ref{mtheorem} for the two-species RVM. For this, we repeat Section \ref{ModelCase} -- Section \ref{ProofofMainTheorem}. In particular, for the trajectories of the electrons, we have the following lemma in place of Lemma \ref{t1t2mapping-model}, which gives the definition of the "reflection point" $t^*$ for each $x \in \supp B_{ext, N}$, $t \in [0, T]$, $v_1 \in \mathbb{R}$, $v_2 \in \mathbb{R}$.



\begin{lemma} \label{t1t2mapping-model-electron}
Fix $\epsilon \in (0, 1)$ and let $x \in (0, \frac{1}{N} - \frac{1}{N} \epsilon] \cup [ 1- (\frac{1}{N} - \frac{1}{N} \epsilon), 1)$.
Let $t \in [0, T]$, $v_1 \in \mathbb{R}$, $v_2 \in \mathbb{R}$. Let $(X_N, V_{1,N}, V_{2,N})$ be a trajectory that takes the value $(x, v_1, v_2)$ at time $t$, given by the ODEs below (that is, the trajectory equations for electrons when the internal electromagnetic fields are removed):
\begin{equation}
\label{E:N-mock-electron}
\left\{
\begin{aligned}
&\dot X_N = \hat V_{1,N} \\
&\dot {V}_{1,N} =  - \hat {V}_{2,N} B_{ext, N} (X_N)   \\
&\dot {V}_{2,N} =  + \hat {V}_{1,N} B_{ext, N} (X_N)  
\end{aligned}
\right.
\end{equation}
Let $I_0 = I_0 (t, x, v_1, v_2)$ be the maximal time interval that contains $t$, and on which $X_N (s; t, x, v_1, v_2)$ lies in $\supp B_{ext, N}$. Then for any fixed $t$, $x$, $v_1$, $v_2$, there exists a unique $t^*$ in the same interval $I_0$ such that 
\begin{equation} \label{E:t1t2mapping-model-eq0-electron}
(X_N, V_{1,N}, V_{2,N}) (t;t,x,v_1, v_2) = (X_N, -V_{1,N}, V_{2,N}) (t^*;t,x,v_1, v_2) = ( x, v_1, v_2) .
\end{equation}
Moreover, $t^* -t$ only depends on $(x, v_1, v_2)$ and $|t^* -t| \lesssim  \frac{1}{N \epsilon}$. 
For any fixed $(x, v_1, v_2)$, $t \mapsto t^*$ as a function of $t$ is $C^\infty$ and invertible. The Jacobian of the inverse mapping $t^*  \mapsto t$ is $|J_N |= | \frac{\partial t}{\partial t^*} | = |J_N (x, v_1, v_2)| =1$.
\end{lemma}

\begin{flushleft}
\textit{Remark} We call $t^*$ the \emph{reflection time} (for the electrons) corresponding to $(t, x, v_1, v_2)$. Notice that Lemma \ref{t1t2mapping-model-electron} only concerns about the behavior of the trajectory of the electron on $I_0$.  
\end{flushleft}

\begin{proof}
The proof is very similar to the one for Lemma \ref{t1t2mapping-model}. We sketch it here. 

Dropping the $N$ subscript for $(X_N, V_{1,N}, V_{2,N})$ in this lemma and passing to polar coordinates for $V$:
$$V_1 = R \cos \Phi \,, \qquad V_2 = R \sin \Phi$$
As in Lemma \ref{t1t2mapping-model}, 
we verify that $R$ is constant on $I_0$ and we find that \eqref{E:N-mock-electron} becomes
\begin{equation}
\label{E:mock-2-electron}
\left\{ 
\begin{aligned}
&\dot{X} = \frac{R \cos \Phi}{\sqrt{1+R^2}}\\
&\dot{\Phi} = + \frac{1}{\sqrt{1+R^2}} N  \Psi' (NX)
\end{aligned}
\right.
\end{equation}

Let us consider the boundary $x=0$ and recall $ \Psi'(Y) \lesssim -\epsilon  < 0$ for $Y = NX \in (0, 1- \epsilon]$. We have $\dot{\Phi}  < 0$.  


Since the trajectory is in $C^1$ and $\dot{\Phi}  < 0$ when $s \in I_0$, $\Phi(s)$ evolves in the direction of decreasing angle. Let us discuss first the case when $V_1(t)< 0$ (that is, $\Phi (t) \in (\pi/2, 3\pi/2)$). Let $s_{\text{turn}} :=   \min \{ s>t : \Phi (s) = \frac{3\pi}{2}\}$, whose existence is guaranteed by  $\dot{\Phi}  < 0$. Since $\Phi$ keeps decreasing, $s_{\text{turn}}$ is the unique time in $I_0$ such that $\Phi (s_{\text{turn}}) = 3\pi/2$, $V_1 (s_{\text{turn}}) = 0$, and hence $X$ reaches its minimum at $s = s_{\text{turn}}$. Continuing after $s_{\text{turn}}$, again due to $\dot{\Phi}   < 0$, there exists a unique $t^*$ defined by
$$t^* := \min \{ s>t : \Phi (s) =  3\pi- \Phi (t)\} $$ 
in $I_0$ such that $ \Phi (t^*) = 3\pi-\Phi(t)$. This gives a unique $t^*$ in the interval $I_0$ such that $ (V_{1}, V_{2}) (t;t,x,v_1, v_2) =  (-V_{1}, V_{2}) (t^*;t,x,v_1, v_2) $. Here we used the fact that $\sqrt{V_{1} (s)^2 + V_{2}(s)^2}  =  R(s) \equiv const.$ on $I_0$. After the time $t^*$, $ V_1  (s;t,x,v_1, v_2)  > 0$. Notice that the region $\supp_x B_{ext, N}$ is of size $O(\frac{1}{N})$, which tells us that when $N$ is large enough, the trajectory $(X  , V_1 , V_2  ) (s;t,x,v_1, v_2)$ can exit $\supp_x B_{ext, N}$ within a time period of order $O(\frac{1}{N})$. Therefore, $V_1 (s;t,x,v_1, v_2)$ can only change its sign once in $I_0$.




\begin{center}
\includegraphics[width=1\textwidth]{Trajectory_1.png}
\end{center}






Arguing similarly as in Lemma \ref{t1t2mapping-model}, we can obtain 
$$ (X, V_{1}, V_{2}) (t;t,x,v_1, v_2) = (X, -V_{1}, V_{2}) (t^*;t,x,v_1, v_2) $$
and 
\begin{equation*}
|t^* - t| \lesssim \frac{1}{N} \sqrt{1+(k_0 + C_1 T)^2} \cdot 2 \pi  \frac{1}{\epsilon} = \frac{2 \pi}{N \epsilon} \sqrt{1+(k_0 + C_1 T)^2} \lesssim \frac{1}{N \epsilon} .
\end{equation*}
Here we make use of the fact that $ \Psi'(Y) < 0$, $  |\Psi'(Y) | \gtrsim \epsilon$ for $Y = NX \in (0, 1 -\epsilon]$ and $|V(s)| \leq k_0 + C_1 T$ (Lemma \ref{NS3-lm3.1}).



For the case $V_1 (t) > 0$ we define
$$t^* :=  \max \{ s<t : \Phi (s) =  3\pi- \Phi (t)\} $$
The case $V_1 (t) =0$ is trivial: We simply take $t^* = t$. For the boundary $x =1$ (that is, $x \in [1- (\frac{1}{N} - \frac{1}{N} \epsilon) , 1)$), the mapping $t \mapsto t^*$ is defined similarly, making use of
$$ \psi_{ext, N} (x) = \Psi (N(1-x))  $$
for $x$ close to $1$. 

To summarize, we define $t^*$ as
$$t^*  := \min \{ s>t : \Phi (s) =  3\pi- \Phi (t)\} $$ 
when $V_{1}(t) < 0$, $x \in (0, \frac{1}{N} \epsilon]$ or $V_{1}(t) > 0$, $x \in [1- (\frac{1}{N} - \frac{1}{N} \epsilon), 1)$, and 
$$t^*  :=  \max \{ s<t : \Phi (s) =  3\pi- \Phi (t)\} $$
when $V_{1}(t) > 0$, $x \in (0, \frac{1}{N} \epsilon]$ or $V_{1}(t) < 0$, $x \in [1- (\frac{1}{N} - \frac{1}{N} \epsilon), 1)$. By the same argument as in Lemma \ref{t1t2mapping-model}, we verify that \eqref{E:t1t2mapping-model-eq0-electron} is satisfied, and moreover, $t^* -t$ only depends on $(x, v_1, v_2)$ and $|t^* -t| \lesssim  \frac{1}{N \epsilon}$. 
For any fixed $(x, v_1, v_2)$, $t \mapsto t^*$ as a function of $t$ is $C^\infty$ and invertible. The Jacobian of the inverse mapping $ t^*  \mapsto t$ is $|J_N |= | \frac{\partial t}{\partial t^*} | = |J_N (x, v_1, v_2)| =1$.

\end{proof}

From Lemma \ref{t1t2mapping-model-electron}, we track the (electron) trajectory backwards in time as in Corollary \ref{cor-t1t2mapping-model} and obtain

\begin{corollary} \label{cor-t1t2mapping-model-electron}
Let $t$, $x$, $v_1$, $v_2$ be as in Lemma \ref{t1t2mapping-model-electron}, then
\begin{equation} \label{E:cor-t1t2mapping-model-eq0-electron}
(X_N, V_{1,N}, V_{2,N}) (0;t,x,v_1, v_2) = (X_N, V_{1,N}, V_{2,N}) (0;t^*,x,-v_1, v_2) . 
\end{equation}
\end{corollary}


In the end, repeating the argument of verifying the Vlasov equation for $f^+$ and $f^-$ separately and the argument of verifying the Maxwell equations for $(E, B)$ in Section \ref{GeneralCase} and Section \ref{ProofofMainTheorem}, we obtain:

\begin{theorem} \label{mtheorem-TwoSpecies}
For each $N$, we consider a $C^1$ solution $(f^{N,\pm}, E^N_1, E^N_2, B^N)$ on $[0, T]$ to \eqref{VlasovBext-TwoSpecies}, \eqref{MaxwellBext}, with the initial-boundary conditions \eqref{boundarycondBext-TwoSpecies}. There exists a subsequence of $(f^{N,\pm} , E^N_1 , E^N_2 , B^N)$, such that $f^{N,\pm} \rightharpoonup f^\pm \ weakly^*$ in $L^\infty ([0, T] \times \Omega \times \mathbb{R}^2)$, $(E^N_1, E^N_2, B^N) \rightarrow (E_1, E_2, B)$ strongly in $ C^0  ([0, T] \times \Omega ) $. The limit $(f^\pm , E_1, E_2, B)$ is a weak solution for \eqref{VlasovBextspecular-TwoSpecies}, \eqref{MaxwellBextspecular}, with exactly the same initial and boundary conditions \eqref{boundarycondBext-TwoSpecies} and the specular boundary condition \eqref{specularBC-TwoSpecies} on $[0, T]$.
\end{theorem}

\section{Appendix}

For the readers' convenience, we introduce the following lemma on ODE perturbation theory:

\begin{lemma}[Gronwall]
\label{L:gronwall}
Suppose $X:\mathbb{R} \to \mathbb{R}^d$ solve
\begin{align*}
\dot X(s) = f(X(s)) +  g(X(s),s)  
\end{align*}
with the initial conditions satisfy $X(0)  = x_0$, $f: \mathbb{R}^d \to \mathbb{R}^d$, $ g: \mathbb{R}^{d+1} \to \mathbb{R}^d$, $|g(X,s)| \leq K$. Assume that $X$, $f$ are differentiable and $g$ is continuous. Suppose that the $d\times d$ matrix $f'(X)$ is uniformly bounded: for all $X\in \mathbb{R}^d$, 
$$  \| f'(X)\|_{L^\infty} \leq \kappa .  $$
Then 
$$|X(s) | \leq  |x_0| + \frac{1}{\kappa} (|f(x_0)| + K) (e^{\kappa s} - 1)  . $$
\end{lemma}

\begin{proof}
We have
\begin{align*}
\frac{d |X(s)- x_0 | }{ds}  \leq |  \dot X(s) |  \leq  \kappa |X(s)-x_0| + |f(x_0)| + | g(X(s) ,s) |   \leq  \kappa |X(s)-x_0| + |f(x_0)| +  K .  
\end{align*} 
The standard integrating factor method completes the proof. 
\end{proof}

\section{Acknowledgement}

The author would like to express her gratitude to her advisor, Professor Walter Strauss, for bringing this research topic to her attention, and also for the invaluable guidance, encouragement and patience, without which this work would be impossible. Also, she thanks Professor Benoit Pausader, Professor Yan Guo and Professor Justin Holmer for helpful discussions. The author is supported by AMS-Simons Foundation for an AMS-Simons Travel Grant. This work does not have any conflicts of interest.

\end{document}